\documentclass[letterpaper,10pt,conference]{IEEEtran}

\usepackage{multirow}
\usepackage{booktabs}
\usepackage{stmaryrd}
\usepackage{amsthm}
\usepackage{cite}
\usepackage{csquotes}
\usepackage{url}
\usepackage{amsmath,amssymb,amsfonts}
\usepackage{algorithmic}
\usepackage{tikz}
\usepackage{graphicx}
\usepackage{textcomp}
\usepackage{xcolor}
\usepackage{enumitem}
\usepackage{hyperref}
\usepackage{makecell}
\usepackage[para]{threeparttable}
\usepackage[ruled,vlined]{algorithm2e}
\usepackage{mathtools}
\usepackage{bm}

\usetikzlibrary{shapes}
\newcommand*\circled[1]{\tikz[baseline=(char.base)]{
  \node[shape=circle,draw,inner sep = 0.3pt] (char) {#1};}}

\def\BibTeX{{\rm B\kern-.05em{\sc i\kern-.025em b}\kern-.08em
    T\kern-.1667em\lower.7ex\hbox{E}\kern-.125emX}}

\newtheorem{definition}{Definition}[subsection]
\newtheorem{theorem}{Theorem}[subsection]

\newcommand\myeq{\mathrel{\overset{\makebox[0pt]{\mbox{\normalfont\tiny\sffamily def}}}{=}}}

\newcommand{\mat}[1]{\bm{#1}}

\DeclarePairedDelimiter{\ct}{\llbracket}{\rrbracket}

\DeclareMathAlphabet{\mathpzc}{OT1}{pzc}{m}{it}

\begin{document}
\title{Biometric Verification Secure Against\\
Malicious Adversaries*
}

\author{
  \IEEEauthorblockN{
    Amina Bassit$^{1}$, Florian Hahn$^{1}$, Joep Peeters$^{1}$, Tom Kevenaar$^{2}$, Raymond Veldhuis$^{1}$, Andreas Peter$^{1}$
  }  
  \IEEEauthorblockA{ 
  \{a.bassit, f.w.hahn, r.n.j.veldhuis, a.peter\}@utwente.nl, \\
  joep@jjpeeters.nl, tom.kevenaar@genkey.com\\
  $^{1}$University of Twente, Enschede, The Netherlands\\
  $^{2}$GenKey Netherlands B.V., Eindhoven, The Netherlands}
}

\maketitle

\begin{abstract}
Biometric verification has been widely deployed in current authentication solutions as it proves the physical presence of individuals.
To protect the sensitive biometric data in such systems, several solutions have been developed that provide security against honest-but-curious (semi-honest) attackers. 
However, in practice attackers typically do not act honestly and multiple studies have shown drastic biometric information leakage in such honest-but-curious solutions when considering dishonest, malicious attackers.

In this paper, we propose a provably secure biometric verification protocol to withstand malicious attackers and prevent biometric data from any sort of leakage.
The proposed protocol is based on a homomorphically encrypted log likelihood-ratio-based (HELR) classifier that supports any biometric modality (e.g. face, fingerprint, dynamic signature, etc.) encoded as a fixed-length real-valued feature vector and performs an accurate and fast biometric recognition.
Our protocol, that is secure against malicious adversaries, is designed from a protocol secure against semi-honest adversaries enhanced by zero-knowledge proofs.
We evaluate both protocols for various security levels and record a sub-second speed (between $0.37$s and $0.88$s) for the protocol against semi-honest adversaries and between $0.95$s and $2.50$s for the protocol secure against malicious adversaries.
\end{abstract}


\section{Introduction}
\renewcommand*{\thefootnote}{\fnsymbol{footnote}}
\footnotetext[1]{This paper is a complete reworking and major expansion of our former paper~\cite{peeters2017fast} (including the addition of further co-authors), which investigated our homomorphically encrypted log likelihood-ratio-based classifier in the semi-honest attacker model only (see Section~\ref{subsection:shprotocol} of the present paper). In the paper at hand, we enhance our protocol from~\cite{peeters2017fast} with tailored zero-knowledge proofs to make it secure against malicious adversaries, which we prove following Canetti’s security model in the two-party case. Furthermore, we evaluate the performance of both protocols (i.e.,~\cite{peeters2017fast} and its enhanced version) on the three different datasets: BMDB, PUT and FRGC.}
\renewcommand*{\thefootnote}{\arabic{footnote}}
Biometric verification plays a pivotal role in current authentication technologies. 
Through measuring biometric modalities, such as faces, biometric verification provides evidence of the physical presence of individuals. 
In comparison to passwords, PIN codes and tokens, biometric data is irreversible and cannot be reissued once leaked or compromised. 
This categorizes it as highly sensitive data that is constantly subject to severe security threats.
The major challenges encountered in relation with biometric data are its storage and its processing that tend to be performed in an unprotected manner. 
Real-life examples confirm the seriousness of these security threats. 
On August 2019, \cite{biostarbreach} reported a biometric data breach in the security platform BioStar2 exposing facial recognition data and fingerprint data of millions of users.
On November 2020, \cite{tronicsXchangebreach} reported another biometric data breach in TronicsXchange's AWS S3 Bucket that was left unprotected, leaking approximately $10.000$ fingerprints. 
These incidents show the urgency of protecting biometric data that is Personally Identifiable Information (PII).
At the same time, many countries have legislations (e.g. the EU's GDPR) that govern how PII of their citizens should be handled, which includes the use of strong data protection technologies.

Biometric verification systems (e.g. multi-user access control) involve two protocols: enrollment and verification that include users, a client and a server as main entities.
The client represents the acquisition device, such as a biometric scanner. 
Its role is to capture the user's biometric reference data during the enrollment and the live \emph{probe} during the verification. 
The server, on the other hand, stores the biometric reference data together with some auxillary information in a \emph{template} during the enrollment and compares it with the live probe during the verification. 
The aim of protecting the biometric data throughout the entire verification process implies secure storage and secure processing which is achievable via homomorphic encryption.
On the one hand, homomorphic encryption offers flexibility in manipulating encrypted data without decryption but on the other hand this same flexibility makes tracking the computations a very difficult task especially when the parties may not be trusted.

From a security point of view in the context of biometrics, the client or the server could be compromised by an attacker who tries not to bypass the authentication but to leak sensitive biometric data that is either stored (the template) or freshly captured (the probe).
In an attempt to remain unnoticed, this attacker could either try to follow the biometric verification protocol as intended and use any information gained from the protocol execution to infer knowledge about a template or probe; or he could arbitrarily deviate from the protocol by following a specific strategy to achieve his desired adversarial goal. 
The security literature \cite{goldreich2009foundations} describes the first type as a \emph{semi-honest} attacker and the second type as a \emph{malicious} attacker.
The protection against malicious (resp. semi-honest) attackers is considered as the strongest (resp. basic) achievable security level.

The security of state-of-the-art biometric verification systems can be split in two categories: \emph{semi-honest client and semi-honest server} \cite{upmanyu2010blind, yasuda2013packed, chun2014outsourceable, cheon2016ghostshell, im2016privacy} and \emph{malicious client and semi-honest server} \cite{shahandashti2015reconciling, vsedvenka2014secure, gunasinghe2017privbiomtauth, im2020practical}. 
Unfortunately, all these systems show drastic biometric information leakage when considering a malicious server as described in~\cite{simoens2012framework}.  
For instance, in the case of the system studied in~\cite{abidin2014security}, a malicious server can send encrypted computations of his own choice instead of the ones dictated by the protocol.
Although the studied verification protocols in \cite{abidin2014security} employ encryption schemes based on the ring-LWE problem, this attack enables a server to learn the biometric template in at most $2N-\theta$ queries (where $N$ the bit-length of a biometric template and $\theta$ the probe-template comparison threshold). 
Both \cite{simoens2012framework} and \cite{abidin2014security} emphasize that a biometric verification assuming a semi-honest server or client puts the biometric data in peril.

Thus far, THRIVE~\cite{karabat2015thrive} is the only work that tried to address this problem by studying the case of both malicious client and malicious server. 
While the authors made some steps forward in this direction, the resulting protocol is unfortunately not secure in the malicious model as their proof follows the definition of the semi-honest model only and does not capture malicious attack strategies.
In Section~\ref{section:RelatedWork}, we elaborate on why THRIVE does not achieve security against malicious adversaries.
Further, we demonstrate concrete attacks against THRIVE achieving security violations in the malicious model. 
Hence, the problem of achieving a solution that is secure against a malicious client and a malicious server remains open.

For the sake of practicality, biometric verification protocols require a verification time limited to a few seconds. 
While existing systems achieve this requirement only for the elementary security level against semi-honest attackers, it is vital to maintain this runtime in the order of seconds also for stronger security levels.

In this paper we confront the security-efficiency trade-off by presenting a practical biometric verification protocol that achieves security in the malicious model.  
We adopt the data-driven biometric recognition approach based on the log likelihood ratio (LLR) classifier \cite{bazen2004likelihood}, known for its optimality in the Neyman-Pearson sense. 
Our approach, called \emph{homomorphically encrypted log likelihood-ratio-based} (HELR) classifier, allows us to speed up the biometric recognition by pre-computing the classifier and storing it into lookup tables. 
Thus when applying an encryption layer, the recognition performance does not degrade in comparison to the unprotected classifier. 
Our HELR classifier supports any biometric modality that can be encoded as a fixed-length real-valued feature vector (such as faces) and does not support the one encoded as a binary-valued feature vector such as irises.
Based on our accurate (EER between $0.25\%$ and $0.27\%$ for faces) HELR classifier, we first present a fast (between $0.37$s and $0.88$s depending on the desired bit security level) biometric verification protocol secure against a semi-honest client and server.
Then we address the above-mentioned problem by proposing a practical (between $0.95$s and $2.50$s depending on the desired bit security level) biometric verification protocol secure against both a \emph{malicious client and malicious server}.
The template is encrypted using threshold homomorphic encryption (THE) such that neither the client nor the server can decrypt it on its own. 
The probe is encrypted by the client only using homomorphic encryption (HE). 
As a matter of fact, encryption alone does not guarantee neither computation correctness nor security in the presence of malicious adversaries. 
Therefore, we force the client and the server to follow the protocol steps of our semi-honest construction by using zero-knowledge proofs (ZK-proofs) to check and keep track of the computations. 
To realize this, integer-oriented THE and HE schemes with compatible ZK-proofs are required. 
We use the additive homomorphic ElGamal encryption scheme and adapt three sigma protocols to suit our construction. 
The proposed protocol is secure and protects the biometric information from any sort of leakage in the presence of malicious adversaries; and imposes to both the client and server to follow the protocol honestly; if one of them tries to misbehave the other entity will detect it and terminate the protocol. 
The latter prohibits sabotaging the protocol by any malicious attacker and also distinguishes between such attempt and a \emph{no match}.

In summary, we make the following contributions:
\begin{itemize}
  \item We introduce the HELR lookup tables that speed up and simplify the biometric recognition reducing it to three elementary operations (i.e. selection, addition and comparison) paving the way for applying an encryption layer over these operations without degrading the biometric accuracy.  
  \item We design two biometric verification protocols that perform the recognition under encryption preventing biometric information (i.e. template, probe and score) from leakage in the presence of semi-honest and malicious adversaries. 
  \item We prove the security of our protocols, evaluate their computational performance and show that we achieve high efficiency while maintaining a strong security level.  
\end{itemize}

\section{Preliminaries}

In this work, we denote by $\mat{x} =\left(x_1, \cdots, x_k\right)$ a $k-$dimensional feature vector, $ \mat{X} = (X_1, \cdots, X_k)$ its corresponding multivariate random variable (from which the features are sampled), $f_{\mat{X}}$ its corresponding probability density function (PDF) and $f_{\mat{X},\mat{Y}}$ the joint probability density function of $\mat{X}$ and $\mat{Y}$.

\subsection{Biometric Background}\label{subsection:bioBackground}

\textbf{Overview:} Biometric verification systems check the authenticity of a claimed user identity by exploiting the fact that biometric traits discriminatively characterize individuals. 
Figure~\ref{fig:biometric-system} depicts the main phases of those systems.
The captured biometric raw measurement goes through a feature extraction step to yield a feature vector. 
During the enrollment phase, the extracted feature vector represents the \emph{template} and is stored along with the user's identity in the system's database.
Later in the verification phase, a user claims to have a certain identity.
The extracted feature vector, in this phase, is called a \emph{probe}.
Subsequently, the system compares the template obtained in the enrollment with the probe by measuring the similarity, in terms of a score, between both feature vectors.
In case the score exceeds a preset threshold $\theta$, the system considers the user genuine and outputs \emph{match} otherwise it considers the user an impostor and outputs \emph{no match}.
\begin{figure}[!h]
  \centering
   \includegraphics[width=\columnwidth]{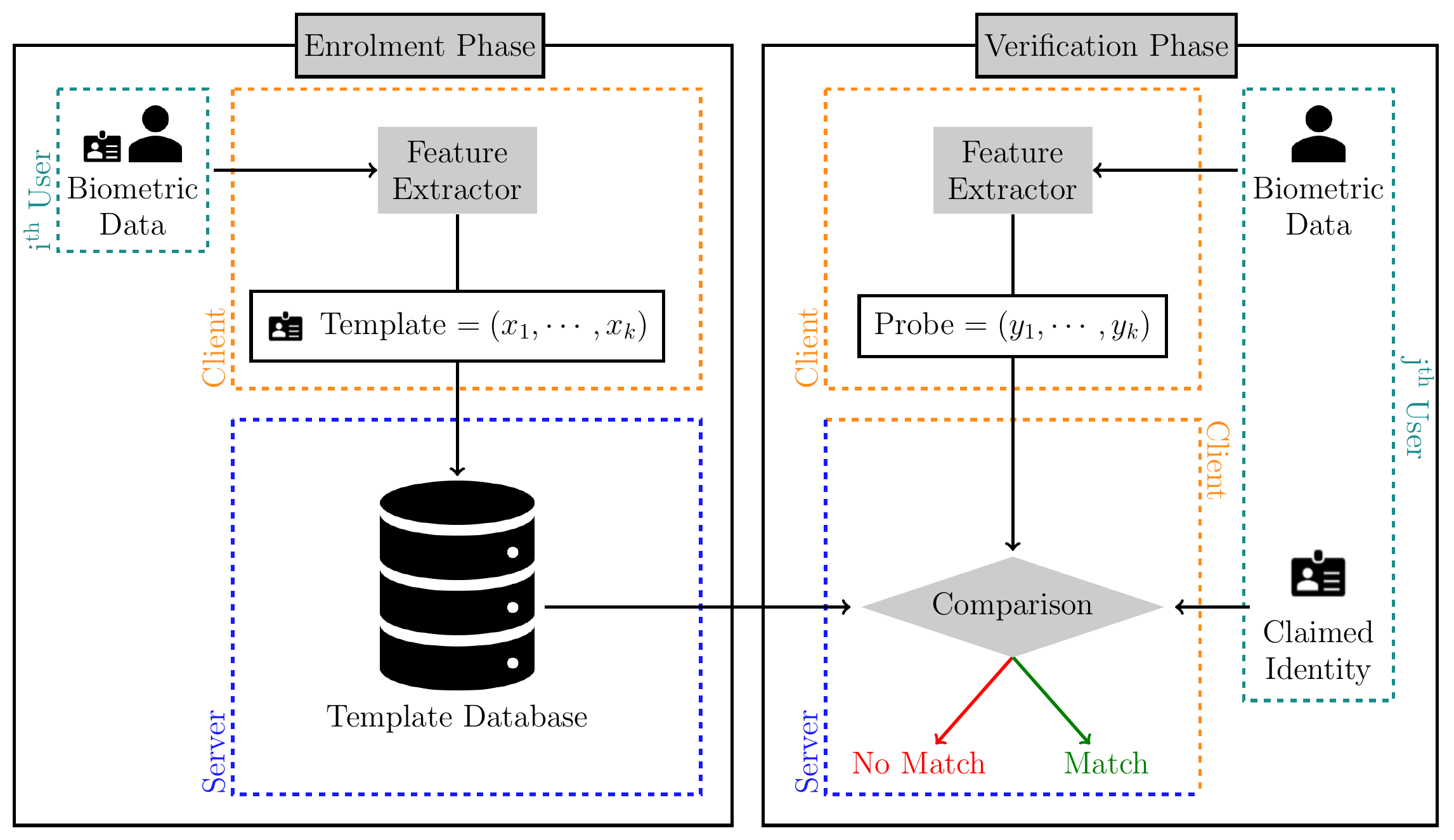}
  \caption{Overview of a multi-user biometric verification system where the comparison is performed between the client and the server. Note that $(x_1, \cdots, x_k)$ and $(y_1, \cdots, y_k)$ represent the feature vectors of the template and a probe.}
  \label{fig:biometric-system}
\end{figure}

\textbf{Log Likelihood Ratio Classifier:} In \cite{bazen2004likelihood}, the authors show the optimality in the Neyman-Pearson sense of the log likelihood ratio as similarity measure comparing two fixed-length feature vectors (one representing the template and the other representing the probe).
The decision comparison is made based on the hypothesis that the system is dealing with the same person (i.e. \emph{genuine verification}) versus the hypothesis that it is dealing with a different person (i.e. \emph{impostor verification}). 
Consider $\mat{x}$ (resp. $\mat{y}$) a biometric feature vector from the enrollment (resp. verification) and $\mat{X}$ (resp. $\mat{Y}$) its corresponding multivariate random variable where their $X_i$ (resp. $Y_i$) are assumed to be independent and normally distributed. 
For the i$^{\text{th}}$ feature, the distribution of the genuine verification is defined by
\begin{equation}\label{eq:gendist}
  P_{X_{i}, Y_{i}}(x_i,y_i|\text{gen}) = f_{X_{i}, Y_{i}}(x_i,y_i)
\end{equation}
a cigar-shaped 2D Gaussian distribution (see Figure~\ref{fig:helr}), whereas the distribution of the impostor verification is defined by 
\begin{equation}\label{eq:impdist}
  P_{X_{i}, Y_{i}}(x_i,y_i|\text{imp}) = f_{X_{i}}(x_i) \cdot f_{Y_{i}}(y_i)
\end{equation}
a circular-shaped 2D Gaussian distribution (see Figure~\ref{fig:helr}).
From these distributions, the similarity between $x_i$ and $y_i$ is measured by calculating the log likelihood ratio (LLR) score.
\begin{align}\label{eq:llr}
  \displaystyle s(x_i,y_i) 
  &= \log \left( \frac{P_{X_{i}, Y_{i}}(x_i,y_i|\text{gen})}{P_{X_{i}, Y_{i}}(x_i,y_i|\text{imp})} \right) 
\end{align}
The LLR classifier is based on the data-driven approach since it requires the knowledge of $f_{X_{i}}$, $ f_{Y_{i}}$ and $f_{X_{i}, Y_{i}}$ that are in practice estimated from a dataset representative of the relevant population.  
Also, this approach assumes that the features were extracted in a statistically independent and identically distributed (i.i.d.) manner.
In practice, this can be achieved by applying a combination of principle component analysis (PCA) and linear discriminant analysis (LDA) as in \cite{bazen2004likelihood} and \cite{10.1007/978-3-642-33712-3_41}.
As a consequence, the final similarity score between two feature vectors $\mat{x}$ and $\mat{y}$ is given by the sum of the individual LLR scores $s(x_i,y_i)$ since the independence between features is assumed.
\begin{equation}\label{eq:finalscore}
  \displaystyle s(\mat{x},\mat{y}) = \sum_{i=1}^{k} s(x_i,y_i)
\end{equation}
The verification system defines a threshold $\theta$ based on which only the final scores that are above $\theta$ are counted as a \emph{match} whereas those below are counted as a \emph{no match}.

\textbf{Performance assessment:} The performance of biometric verification systems is tightly related to the performance of their core comparison algorithm (called comparator) and expressed in terms of False Non-Match Rate (FNMR), that is the probability that the comparator decides no match for two samples coming from the same individual; and False Match Rate (FMR) that is the probability that the comparator decides match for two samples coming from two different individuals.
An infinitely high threshold $\theta$ results in FNMR $=1$ and FMR $=0$, lowering the threshold decreases (increases) the FNMR (FMR), respectively towards an infinitely low threshold $\theta$ for which FNMR $=0$ and FMR $=1$.
This trade-off can be graphically illustrated by a decision error trade-off (DET) curve representing the FNMR as a function of the FMR. 
The Equal Error Rate (EER) denotes the point on the curve where FNMR and FMR are equal.
FNMR@FMR $= 0.1\%$ denotes the point on the curve where FMR $= 0.1\%$.
The system's threshold $\theta$ is set to meet an amount of acceptable FMR;
often at FMR $= 1\%$ or $0.1\%$.

\subsection{Additively Homomorphic ElGamal}\label{subsection:elgamal}
We briefly recall the additively homomorphic ElGamal and its $(2,2)$-threshold version \cite{10.5555/1754542.1754554}.
Let $q$ be a large prime, $\mathbb{G}$ a group of order $q$ and generator $g$.
Let $k = g^s$ be the public key corresponding to the private key $s$.
The encryption of a message $m \in \mathbb{Z}_q$ is $[m]\myeq (g^r,g^m \cdot k^r)$ where $r\in\mathbb{Z}_q$ is random.
The decryption of the ciphertext $[m]$ is the discrete log of $g^m \cdot k^r\cdot (g^r)^{-s}$.
The additively homomorphic ElGamal is secure against Indistinguishable Chosen-Plaintext Attack (IND-CPA) \cite{10.5555/1754542.1754554} under the Decisional Diffie-Hellman (DDH) assumption.
It supports the following operations: ciphertext multiplication $\left[ m_1 \right] \cdot \left[ m_2 \right] = \left[ m_1 + m_2 \right] $, re-randomizing with randomness $r_0$: $\left[ m \right] \cdot \left[ 0 \right] = \left( g^{r+r_0}, g^{m+0} \cdot k^{r+r_0} \right) = \left[ m \right] $, blinding with blinding value $r_d$: $\left[ m \right]^{r_d} = \left( g^{r_d \cdot r }, g^{r_d \cdot m} \cdot k^{r_d \cdot r } \right) = \left[ r_d \cdot m \right]$ and subtraction of two ciphertexts $\left[ m_1 \right] - \left[ m_2 \right] \myeq \left[ m_1 \right] \cdot \left[ m_2 \right]^{-1} = \left[ m_1 - m_2 \right]$.

For implementing a $(2,2)$-threshold additively homomorphic ElGamal, we use the technique described in \cite{pedersen1991threshold}.
Given the key pair $(pk_i,sk_i)$ for $i\in\left\{1,2\right\}$ such that $pk_i = g^{sk_i}$; $pk_{joint}\myeq pk_1 \cdot pk_2$ is the joint public key.
The encryption is performed under $pk_{joint}$ and denoted as $\ct{m}\myeq (g^r,g^m \cdot pk_{joint}^r)$.
The decryption of $\ct{m}$ comes into two stages.
First, each party $i$, using its private key $sk_i$, produces a partial decryption $[m]_{i}\myeq (g^r, g^m \cdot pk_{joint}^r\cdot (g^r)^{-sk_i})$.
Then by combining the exchanged partial decryptions, the final decryption is the discrete log of $g^m \cdot pk_{joint}^r \cdot (g^r)^{-sk_1} \cdot (g^r)^{-sk_2}$.
Notice here that a partial decryption $[m]_{1}$ is a non-threshold ElGamal ciphertext encrypted under $pk_2$ and vice versa for $[m]_{2}$.
Threshold ElGamal is also IND-CPA secure \cite{fouque2001threshold} under the DDH assumption and supports all the above-mentioned operations.

\subsection{Zero-knowledge proofs}
ZK-proofs allow to prove a statement without revealing the secret.
In the literature \cite{hazay2010efficient}, ZK-proofs constructed from $\Sigma\text{-protocols}$ are efficient and flexible to fit the desired proof. 
It is also possible to combine them to prove the conjunction of several statements (called AND proofs). 
In Table~\ref{tab:sigmaplain}, we recall, from \cite{camenisch1998group}, the $\Sigma\text{-protocol}$, denoted as $\Sigma_{\text{Plain}}$, that proves the plaintext knowledge of an ElGamal ciphertext $[m]=(u,v)$.
It can be enhanced by the generic construction in \cite{hazay2010efficient} to transform it into a zero-knowledge proof of knowledge; that we denote as $\text{ZKPoK}_{\text{Plain}}$ and use later in Protocol Figure~\ref{fig:mal-biometric-verification}.
In the same protocol, we also use non-interactive ZK-proofs constructed from $\Sigma\text{-protocols}$ using the Fiat-Shamir transformation \cite{fiat1986prove}.
\begin{table}[h!]
    \centering
    \caption{$\Sigma_{\text{Plain}}$ protocol that proves the plaintext knowledge of an ElGamal ciphertext $[m]=(u,v)$} \label{tab:sigmaplain}
    \begin{tabular}{l|l|l|l}
      \toprule 
      \textbf{Commitment} & \textbf{Challenge} & \textbf{Response}&\textbf{Verification}\\
      \midrule 
      $U = g^{r'}$ & \multirow{2}{*}{$e\in \left\{0,1\right\}^t $} & $z_r= r'+e \cdot r $ & $g^{z_r} =^{?} U \cdot u^{e}$\\
      $ V = g^{m'} \cdot k^{r'}$ & & $z_m = m'+e \cdot m$ & $g^{z_m} =^{?} V \cdot v^{e}$\\
      \bottomrule 
    \end{tabular}
\end{table}

\section{Security model}\label{section:Secumodel}
In this work, we follow Canetti's security model \cite{canetti2000security} for malicious static adversaries in the special case of two parties.
We use the notations and extensions from \cite{cramer2000multiparty} where each party $P_i$ receives a secret input $x^{(s)}_{i}$ and a public input $x^{(p)}_{i}$ and returns a secret output $y^{(s)}_{i}$ and a public output $y^{(p)}_{i}$. 
Also, the adversary receives the public input and output of all parties.

\textbf{Real-world model} Let $\pi$ be a two-party protocol.
Let $ \vec{x} = \left( x^{(s)}_{1}, x^{(p)}_{1}, x^{(s)}_{2}, x^{(p)}_{2} \right)$ be the parties' inputs, $\vec{r} = \left( r_{1}, r_{2}, r_{\mathcal{A}} \right)$ be the parties' and the adversary's $\mathcal{A}$ random inputs and $a \in \{0,1 \}^*$ be the adversary's auxiliary input.
We assume that only one party, $P_j$ where $j\in \{1,2\}$, is corrupted at the time.
$\text{ADVR}_{\pi, \mathcal{A}} \left( k, \vec{x}, \{j\} , a, \vec{r}  \right)$ denotes the output of the adversary and $\text{EXEC}_{\pi, \mathcal{A}} \left( k, \vec{x}, \{j\} , a, \vec{r}  \right)_i$ the output of the party $P_i$ after a real-world execution of $\pi$ in the presence of the adversary $\mathcal{A}$ corrupting the party $P_j$.
\begin{align*}
  \text{EXEC}_{\pi, \mathcal{A}} \left( k, \vec{x}, \{j\} ,  a, \vec{r}  \right) = (
  \text{ADVR}_{\pi, \mathcal{A}} &\left( k, \vec{x}, \{j\} , a, \vec{r}  \right),\\
  \text{EXEC}_{\pi, \mathcal{A}} &\left( k, \vec{x}, \{j\} , a, \vec{r}  \right)_1, \\
  \text{EXEC}_{\pi, \mathcal{A}} &\left( k, \vec{x}, \{j\} , a, \vec{r}  \right)_2  )
\end{align*}

$\text{EXEC}_{\pi, \mathcal{A}} \left( k, \vec{x}, \{j\} , a \right)$ denotes the random variable 
$\text{EXEC}_{\pi, \mathcal{A}} ( k, \vec{x}, $ 
$ \{j\},a, \vec{r} )$ where $\vec{r}$ is chosen uniformly random and $\text{EXEC}_{\pi, \mathcal{A}} = \left\{ \text{EXEC}_{\pi, \mathcal{A}} \left( k, \vec{x}, \{j\} , a \right) \right\} _{k, \vec{x}, \{j\}, a }$ the distribution ensemble indexed by the security parameter $k\in \mathbb{N}$, the input $\vec{x}$, the corrupted party $P_j$ and the auxiliary input $a$.

\textbf{Ideal model} Let $f$ be a probabilistic two-party function computable in PPT defined as $f \big( k, x^{(s)}_{1}, x^{(p)}_{1}, x^{(s)}_{2}, x^{(p)}_{2}, r \big) = \big( y^{(s)}_{1}, y^{(p)}_{1}, y^{(s)}_{2},$
$ y^{(p)}_{2} \big)$ where $k$ is the security parameter and $r$ is the random input.
In the ideal model, each party $P_i$ sends its input $\left( x^{(s)}_{i}, x^{(p)}_{i} \right)$ to the trusted party $\mathbb{T}$ that computes $f$ on the inputs and a uniformly chosen random $r$ then returns to each party $P_i$ its output $\left(y^{(s)}_{i}, y^{(p)}_{i}\right)$.
Note that the malicious static adversary $\mathcal{S}_{\mathcal{A}}$, operating in the ideal execution, uses the real-world adversary $\mathcal{A}$, which corrupts the party $P_j$, as a subroutine.
At the beginning of the execution, $\mathcal{S}_{\mathcal{A}}$ sees the public values of both parties and the secret values of the corrupted party $P_j$ and also substitutes $P_j$'s input by values of his choice.
Again, we denote by $\text{IDEAL}_{f, \mathcal{S}_{\mathcal{A}}} \left( k, \vec{x}, \{j\} , a, \vec{r}  \right)_i$ the output of the party $P_i$ after an ideal execution of $f$ in the presence of the adversary $\mathcal{S}_{\mathcal{A}}$.
\begin{align*}
  \text{IDEAL}_{f, \mathcal{S}_{\mathcal{A}}} \left( k, \vec{x}, \{j\} ,  a, \vec{r}  \right) = (
  \text{ADVR}_{f, \mathcal{S}_{\mathcal{A}}} &\left( k, \vec{x}, \{j\} , a, \vec{r}  \right),\\
  \text{IDEAL}_{f, \mathcal{S}_{\mathcal{A}}} &\left( k, \vec{x}, \{j\} , a, \vec{r}  \right)_1, \\
  \text{IDEAL}_{f, \mathcal{S}_{\mathcal{A}}} &\left( k, \vec{x}, \{j\} , a, \vec{r}  \right)_2  )
\end{align*}
We denote by $\text{IDEAL}_{f, \mathcal{S}_{\mathcal{A}}} \left( k, \vec{x}, \{j\} , a \right)$ the random variable\\ $\text{IDEAL}_{f, \mathcal{S}_{\mathcal{A}}} \left( k, \vec{x}, \{j\} , a, \vec{r} \right)$ where $\vec{r}$ is chosen uniformly random and $\text{IDEAL}_{f, \mathcal{S}_{\mathcal{A}}} = \left\{ \text{IDEAL}_{f, \mathcal{S}_{\mathcal{A}}} \left( k, \vec{x}, \{j\} , a \right) \right\} _{k, \vec{x}, \{j\}, a }$ the distribution ensemble indexed by the security parameter $k\in \mathbb{N}$, the input $\vec{x}$, the corrupted party $P_j$ and the auxiliary input $a$.

\textbf{Hybrid model} In the $\left(g_1, \cdots, g_l \right)$-Hybrid model the execution of a protocol $\pi$ proceeds as in the real-world model, except that the parties have access to a trusted party $\mathbb{T}$ for evaluating the two-party functions $g_1, \cdots, g_l$.
These ideal evaluations proceeds as in the ideal-model.
As above, we define the following distribution ensemble 
\[
 \text{EXEC}^{g_1, \cdots, g_l}_{\pi, \mathcal{A}} = \left\{ \text{EXEC}^{g_1, \cdots, g_l}_{\pi, \mathcal{A}} \left( k, \vec{x}, \{j\} , a \right) \right\} _{k, \vec{x}, \{j\}, a } 
\]

The security in this model is defined by requiring that a real-world execution or $\left(g_1, \cdots, g_l \right)$-Hybrid execution of a protocol $\pi$ for computing a function $f$ should reveal no more information to the adversary than what the ideal evaluation of $f$ does, namely the output. 

\begin{definition}\label{def:malsec}
  Let $f$ be a two-party function and $\pi$ be a two-party protocol.
  We say that $\pi$ securely evaluates $f$ in the $\left(g_1, \cdots, g_l \right)$-Hybrid model if for any malicious static $\left(g_1, \cdots, g_l \right)$-Hybrid adversary $ \mathcal{A}$ corrupting one party, there exists an adversary $\mathcal{S}_{\mathcal{A}}$ operating in the ideal world such that 
  \[
    \text{IDEAL}_{f, \mathcal{S}_{\mathcal{A}} } \stackrel{c}{\approx}  \text{EXEC}^{g_1, \cdots, g_l}_{\pi, \mathcal{A}}
  \]
  where $\stackrel{c}{\approx}$ means the computational indistinguishability of ensembles, see Definition 3 in \cite{canetti2000security}.
\end{definition}

\section{Pre-computed HELR Classifier}\label{section:helr}
Recall from Section~\ref{subsection:bioBackground} the biometric comparison of two feature vectors using the LLR classifier is a data-driven approach where the parameters of the genuine and impostor distributions are estimated from a given dataset representative of the relevant population. 
For each feature, we draw the PDFs in Equation \eqref{eq:gendist} and Equation \eqref{eq:impdist} as estimated from the training dataset.
Assuming that the extracted features follow the Gaussian distribution, they are rendered i.i.d. by applying a combination of PCA and LDA as shown in \cite{bazen2004likelihood} and \cite{10.1007/978-3-642-33712-3_41}.  
Then we compute the LLR per feature as in Equation \eqref{eq:llr} in order to produce the final score in Equation \eqref{eq:finalscore} for each comparison.

The LLR in Equation \eqref{eq:llr} can be visualized as a function with two inputs (the i$^{\text{th}}$ feature, one from the template and one from a probe) and one output (individual score).
This function can be arranged into a lookup table where the rows' indexes represent the possible values of the first input (features from the template), the columns' indexes represent the possible values of the second input (features from a probe) and the cells contain the output (individual scores).
In order to produce such lookup table, a mapping from a continuous domain to a finite set is needed to limit the possible feature values allowing the storage of a representative score per cell.
This respective score is then quantized to an integer in order to facilitate the application of homomorphic encryption. 
An example of generating the HELR lookup table of one feature is given in Figure~\ref{fig:helr}.

\paragraph{Feature Quantization}\label{par:fq}
We describe the feature quantization procedure for the i$^{\text{th}}$ feature, the same is applied for the remaining features.
Assuming the PDFs are zero-mean, which is achievable by subtracting the mean, implies that the impostor PDF has an unit variance.
Recall, that $X_i$ and $Y_i$ are normally distributed, hence we get $X_i \sim \mathcal{N}(0,1)$ and $Y_i \sim \mathcal{N}(0,1)$.
To perform a feature quantization on $n$ levels (called \emph{feature level}) we divide the 2D impostor PDF in an equiprobable manner so that all bins will have the same probability, thus an arbitrary feature observation is just as likely to land on any of those bins.
This is done by determining the bins' borders following Algorithm~\ref{alg:bins} where $\text{ICDF}(p,0,1)$ is the inverse cumulative distribution function of a $\mathcal{N}(0,1)$ at the cumulative probability $p$ and it returns the value associated with $p$. 
\begin{algorithm}[!h]
  \SetAlgoLined
  \KwIn{$n$ feature quantization level }
  \KwOut{$Bn$ array containing the bins' borders}
    $Bn$ array of size $n-1$\;
   \For{$j\leftarrow 1$ \KwTo $n-1$}{
     $p = j/n$ \;
     $Bn[j] = \text{ICDF}(p,0,1)$\;   
   }
   \caption{Procedure to determine the bins' borders}
   \label{alg:bins}
\end{algorithm}

$a_i\in Bn_{a_i}$ (resp. $b_i\in Bn_{b_i}$) denotes the measured value for feature $x_i$ (resp $y_i$) from the first (resp. second) sample and is quantized to $\hat{a_i}$ (resp. $\hat{b_i}$) following Algorithm~\ref{alg:fq} using the same $Bn$ bins' borders array of the i$^{\text{th}}$ feature.
\begin{algorithm}[!h]
  \SetAlgoLined
  \KwIn{$a_i$ raw feature value of the i$^{\text{th}}$ feature and $Bn$ array containing the bins' borders of the i$^{\text{th}}$ feature}
  \KwOut{$\hat{a_i}$ quantized value}
    \For{$j\leftarrow 1$ \KwTo $n-1$}{
     \lIf{$a_i < Bn[j]$}{
      \KwRet{$j-1$}
     }  
   }
   \KwRet{$n-1$}
   \caption{Feature quantization on $n$ feature levels}
   \label{alg:fq}
\end{algorithm}

\paragraph{Score Quantization}\label{par:sq}
As we are dealing with 2D distributions, we partition the impostor PDF and the genuine PDF according to an $n \times n$ grid using $Bn$ for both axes $x_i$ and $y_i$; see the red grid Figure \ref{fig:helr}.
For a cell located at $(\hat{a}_i,\hat{b}_i)$, where $a_i$ and $b_i$ the measured values, we compute the genuine probability distribution (see Equation~\eqref{eq:gendistCell}) inside that cell by calculating the area under the curve delimited by its borders $Bn_{a_i}$ and $ Bn_{b_i}$, see the dotted surface depicted in Figure~\ref{fig:helr}.
Based on Equation~\eqref{eq:gendist} we hence get:
\begin{equation}\label{eq:gendistCell}
  P_{X_{i}, Y_{i}}(\hat{a}_i,\hat{b}_i|\text{gen}) = \int_{Bn_{b_i}} \int_{Bn_{a_i}} f_{X_{i}, Y_{i}} (x_i,y_i) \,dx_{i}\,dy_{i}
\end{equation}
Note that $Bn_{a_i}$ and $Bn_{b_i}$ have one of the three forms $\big]-\infty, Bn[1] \big[, \big[ Bn[j], Bn[j+1] \big[$ or $\big[ Bn[n-1], +\infty \big[$.\\
For the impostor distribution, all cells have identical probability since it was equiprobably divided:
\begin{equation}\label{eq:impdistCell}
  P_{X_{i}, Y_{i}}(\hat{a}_i,\hat{b}_i|\text{imp}) = \frac{1}{n \times n}
\end{equation}
After that, we calculate the LLR for that cell $(\hat{a}_i,\hat{b}_i)$ and place the resulted non-quantized score in a lookup table at row $\hat{a}_i$ and column $\hat{b}_i$.
As described in Section~\ref{subsection:elgamal}, homomorphic encryption requires integers and the resulted LLR scores are real-valued.
We perform a second quantization to map each real-valued non-quantized score to an integer that we call \emph{quantized score} by dividing the real-valued score by a quantization step $\Delta$ and rounding the result to the nearest integer to yield the quantized score $s(\hat{a}_i,\hat{b}_i)$.
\begin{figure}[!h]
  \centering
  \includegraphics[width=\columnwidth]{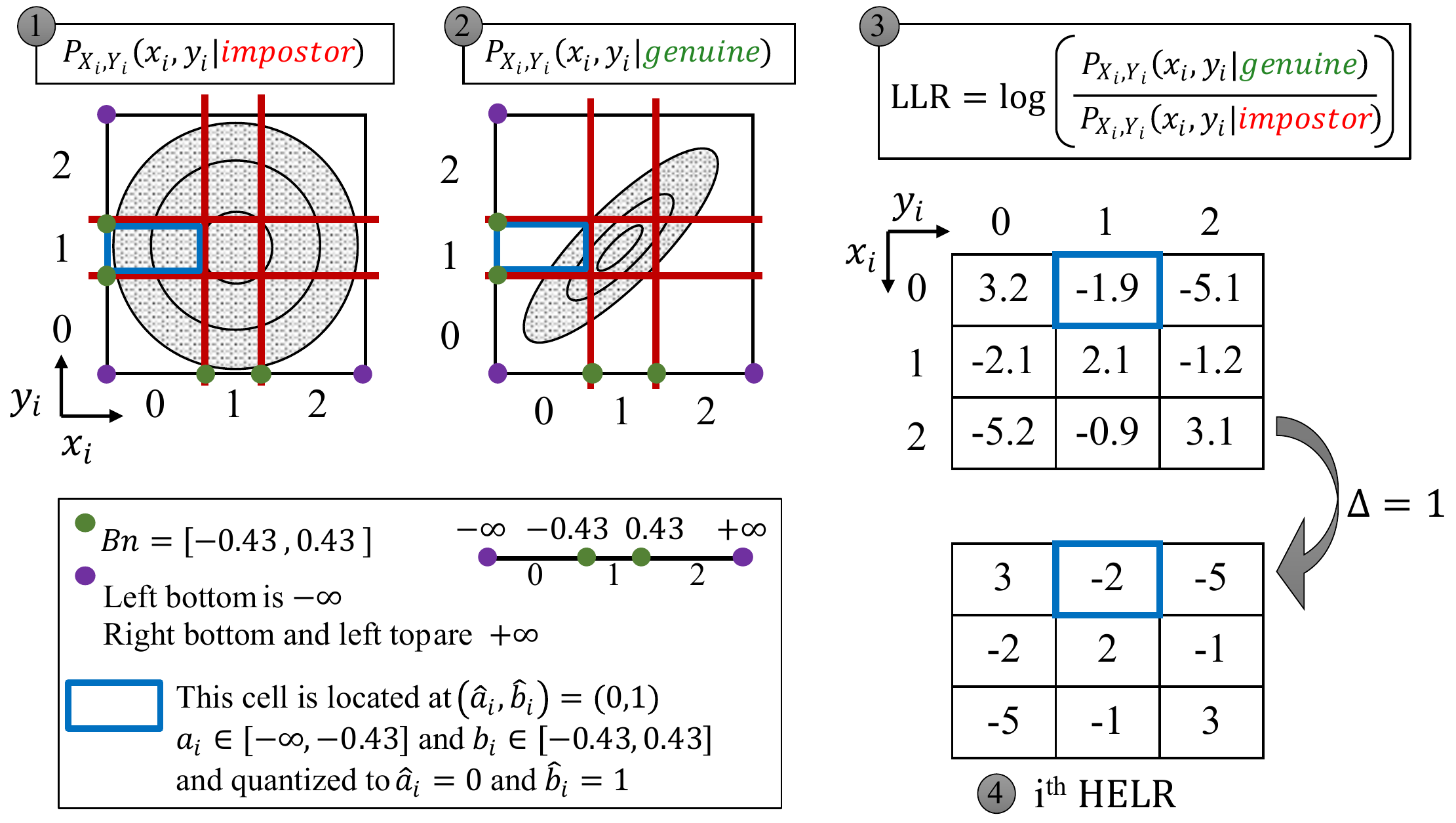}
  \caption{Generation of the i$^{\text{th}}$ HELR lookup table using a feature level $n=3$. First the red $n \times n$ grid equiprobably partitions the impostor PDF with respect to $B_{n}$ along both axes. Then the same grid is applied on the genuine PDF. Subsequently, the LLR is computed then quantized using $\Delta = 1$ and stored in the i$^{\text{th}}$ HELR table. The blue cell is an example of a sample $(a_i,b_i) = (-0.25, 0.31)$. }\label{fig:helr} 
\end{figure}

\paragraph{HELR Lookup tables}
For each feature, we generate an $n \times n$ HELR lookup table where its cells contain the quantized score resulted from a row (resp. column) that refers to the quantized feature value of the first (resp. second) feature vector.
To calculate the similarity score of two feature vectors $\mat{a}=(a_1,\cdots, a_k)$ and $\mat{b}=(b_1,\cdots, b_k)$ using the HELR tables, we map each feature $a_i$ (resp. $b_i$) to its quantized value $\hat{a}_i$ (resp. $\hat{b}_i$) and select its corresponding score $s(\hat{a}_i,\hat{b}_i)$ from the i$^{\text{th}}$ HELR table, the value at location row $\hat{a}_i$ and column $\hat{b}_i$.
Based on Equation~\eqref{eq:finalscore} we calculate the final score as $ S = \sum_{i=1}^{k} s(\hat{a}_i,\hat{b}_i)$. 
Recall, that we can sum the individual scores since the features are assumed to be independent.

Since the HELR lookup tables are generated from a dataset representative of the relevant population, we assume that they are public and accessible by any party namely the client and the server.
The comparison outcome is determined by the final score that is calculated using sensitive biometric data. 
Thus, all values that are involved in this calculation are sensitive. 
Hence, row and column position as well as the individual scores and the final score must be protected.
In the following sections, we aim to perform biometric verification under encryption using the HELR tables where the biometric data is kept encrypted throughout the process and only the comparison outcome (match or no match) is revealed.

\section{Proposed Verification Protocols}\label{section:protocols}

Our final goal is to achieve security against both malicious client and malicious server.
We first design a biometric verification protocol that ensures zero-biometric information leakage secure against semi-honest client and server.
Then, we modify this construction to force them to behave honestly.
Thus, we obtain a protocol secure against malicious client and server ensuring both the correctness of the comparison outcome and zero-biometric information leakage.
In both scenarios, the server must not learn the probe, the unprotected template, the individual scores, the final score and the comparison outcome.
The same requirements apply to the client except for the probe and the comparison outcome that it should be able to learn.

In the following, we suppose that the client and the server respectively hold the key pairs $(pk_{clt}, sk_{clt})$ and $(pk_{ser}, sk_{ser})$ from which the threshold ElGamal public key $pk_{joint}$ is calculated; see Section~\ref{subsection:elgamal}.
The HELR lookup tables, the comparison threshold $\theta$ and the maximum score $S_{\max}$ are public knowledge.
We assume that the initial enrollment process is performed in a fully controlled environment.

\subsection{Protocol Secure Against Semi-Honest Adversaries}\label{subsection:shprotocol}
Prior to a biometric verification, a user should enroll in order to register his template.
In this semi-honest construction, a new user $u\text{ID}$ presents his biometric modality to the client that first extracts a $k$-dimensional feature vector $(f_1, \cdots, f_k)$. 
Then for the $i$-th feature, the client selects the $f_i$-th row from the $i$-th HELR lookup table to form the user's template $\mat{T_{u\text{ID}}}$; which can be seen as a vector of vectors.
\begin{equation}\label{eq:templateSH}
  \mat{T_{u\text{ID}}} = \left(  \left( s_{f_i,j}^{i}\right)_{j\in [1,n]} \right)_{i\in [1,k]}
\end{equation}
Where $s_{f_i,j}^{i}$ is the score at the intersection of row $f_i$ and column $j$ from the HELR lookup table $i$.
Finally, the client encrypts $\mat{T_{u\text{ID}}}$ using $pk_{joint}$ then along with $u\text{ID}$ sends the encrypted template $ \ct{\mat{T_{u\text{ID}}}} $ to the server who stores them for later retrieval.
\begin{equation}\label{eq:encryptedtemplateSH}
  \ct{\mat{T_{u\text{ID}}}} = \left(  \left( \ct{s_{f_i,j}^{i}} \right)_{j\in [1,n]} \right)_{i\in [1,k]}
\end{equation}
Figure~\ref{fig:shverification} describes our verification protocol secure against semi-honest adversaries where the veracity of a user claiming an identity $u\text{ID}$ is assessed.
After extracting a $k$-dimensional feature vector $\mat{P}$ from the acquired live biometric modality, the client requests, from the server, the user's corresponding encrypted template $\ct{\mat{T_{u\text{ID}}}}$.
Recall, from Section~\ref{subsection:elgamal}, that in the additively homomorphic ElGamal, multiplication under encryption is equivalent to addition in the plain domain.
That is, the client selects and multiplies the encrypted individual scores to form the final score $\ct{S}$ then re-randomizes it before sending it back to the server.
This re-randomization prevents the server from guessing the selected individual scores from the encrypted template since the multiplication of two ciphertexts twice yields the same ciphertext.
\begin{figure}[!h]
  \centering
  \includegraphics[width=\columnwidth]{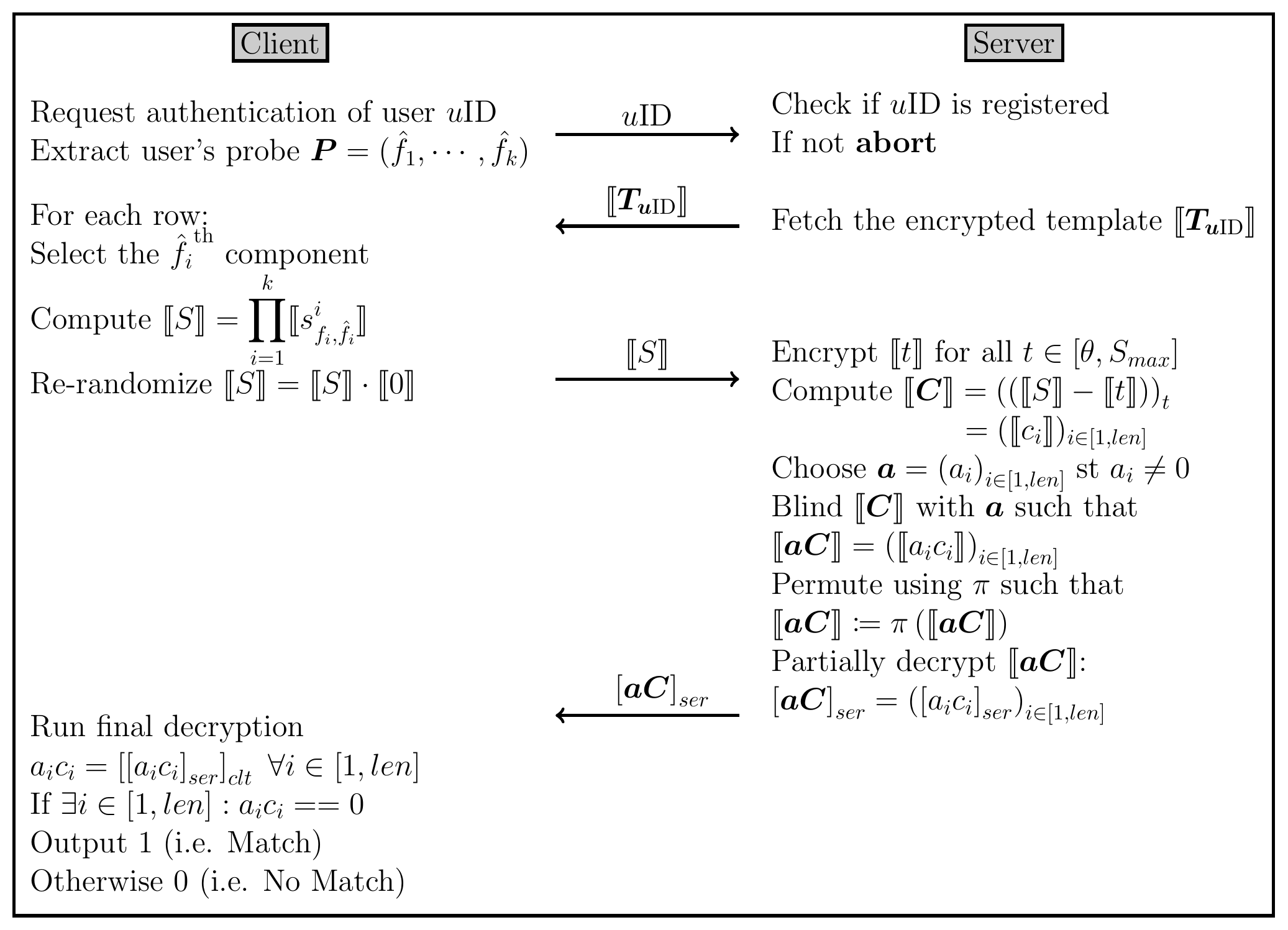}
  \caption{Biometric Verification Protocol Secure Against Semi-Honest Adversaries}
  \label{fig:shverification}
\end{figure}

To determine whether the encrypted final score is above or below the threshold $\theta$, the server needs to perform the comparison under encryption.
One would think of subtracting only $\ct{\theta}$ from $\ct{S}$, however this yields a zero only when both are equal and non-zero value when $S$ is above or below.
To solve this, one should compare $S$ with the integers between $\theta$ and $S_{\max}$. 
If $ \theta \leq S $ then exactly one single subtraction results in $0$, and if $S < \theta$ then all the subtractions are unequal to $0$.
In addition to that, a multiplicative blinding must be applied to these subtractions to protect the final score $S$ when a decryption occurs.    
Note that the multiplicative blinding preserves only the $0$ in contrast with the values unequal to $0$ which become random.

Applying this, the server encrypts the integers between $\theta$ and $S_{\max}$ under $pk_{joint}$ then computes the blinded comparison vector $\ct{\mat{aC}}$ from the comparison vector $\ct{\mat{aC}}$ and $\mat{a}$ a vector of random values.
After that, it applies a random permutation $\pi$ to $\ct{\mat{aC}}$ and partially decrypts the blinded-permuted comparison vector and sends $\left[\mat{ aC }\right]_{ser}$.
The client runs the final decryption on $\left[\mat{ aC }\right]_{ser}$ to retrieve the plain values of the blinded-permuted comparison vector $\mat{aC}$. 
If he finds a value $a_ic_i == 0$ in $\mat{aC}$ that means $S$ is equal or greater than the threshold $\theta$.
In this case, the client outputs \emph{match}. 
If all values in $\mat{aC}$ are different from zero that means $S$ is strictly below the threshold $\theta$.
In this case, the client outputs \emph{no match}.

\textbf{Limitations:}
This verification protocol is suitable for a client and server that trust each other in terms of the correctness of the exchanged messages.
However, in an untrusted setting this protocol leads to serious security threats.
Consider the case of a malicious client that can arbitrarily deviate from the protocol.
Since the system's threshold $\theta$ is public, he can encrypt $\theta$ using the joint public key $pk_{joint}$ to receive a blinded-permuted comparison vector $\mat{aC}$ that contains a zero; thus he succeeds in forcing a \emph{match}.
For the case of a malicious server, instead of sending the real encrypted template, he can craft a template of the form $((0,\cdots,0), \cdots, (1,\cdots, n), \cdots, (0,\cdots,0))$, fixing the individual scores of the i$^{\text{th}}$ feature to $(1,\cdots, n)$ and the individual scores of the remaining features to zero.
He then encrypts the crafted template with his public key $pk_{ser}$ and sends it to the client; who will send back a sum of the encrypted individual scores.
By decrypting the sum, the server learns the value of the i$^{\text{th}}$ component of the probe, repeating this for all $k$ components reveals the probe.
Another attack could be: instead of sending the partial decryption of the permuted comparison vector $\mat{aC}$, it sends a non-zero values vector of the same length as $\mat{aC}$ and encrypted with the client's public key $pk_{clt}$.
Thus he forces a \emph{no match} to a genuine user since the decryption of $\mat{aC}$ yields a non-zero vector.

\subsection{Protocol Secure Against Malicious Adversaries}

To address the limitations mentioned in Section~\ref{subsection:shprotocol}, we transform the construction in Figure~\ref{fig:shverification} into a protocol secure against both malicious client and malicious server using adapted ZK-proofs (see Section~\ref{subsubsection:adaptedZKs}).

\subsubsection{Adapted ZK-proofs from $\Sigma$-protocols}\label{subsubsection:adaptedZKs}
We construct three adapted $\Sigma$-protocols (see Table \ref{tab:adaptedSigma}) that we turn into non-interactive ZK-proofs using Fiat-Shamir transformation (see Table \ref{tab:adaptedZKs}) and provide their proofs in Theorem~\ref{theo:deczero}, Theorem~\ref{theo:blind} and Theorem~\ref{theo:parial}.

\begin{table*}[h!]
  \begin{center}    
    \caption{Adapted $\Sigma$-protocols}\label{tab:adaptedSigma}
    \begin{tabular}{l|l|l|l}
      \toprule 
      & $\Sigma_{\text{DecZero}}$
      & $\Sigma_{\text{Blind}}$
      & $\Sigma_{\text{Partial}}$\\ 
      & for $ [m_1] - [m_2] = (u,v)$
      & for $ [m]^r= (u^r,v^r) = (a,b)$
      & for $[m]_i=(u, v\cdot u^{-sk_i})= (u,c)$\\ 
      \midrule 
      Commitment& $X = g^{r'}, U = u^{r'}$ & $U = u^{r'}, V = v^{r'}$&  $X = g^{r'}, U = u^{r'}$ \\
      \midrule 
      Challenge& $e\in \left\{0,1\right\}^t $ & $e\in \left\{0,1\right\}^t $ & $e\in \left\{0,1\right\}^t $ \\
      \midrule
      Response&  $z= r'+e \cdot s $& $z= r'+e \cdot r $&  $z= r'+e \cdot sk_i $\\
      \midrule
      Verification Equations  & $g^{z} =^{?} X \cdot k^{e}$ & $u^{z} =^{?} U \cdot a^{e}$ & $g^{z} =^{?} X \cdot pk_{i}^{e}$\\
      & $u^{z} =^{?} U \cdot v^{e}$ & $v^{z} =^{?} V \cdot b^{e}$ & $u^{z} =^{?} U \cdot (v \cdot c^{-1})^{e}$\\
      \bottomrule 
    \end{tabular}
  \end{center}
\end{table*}

\begin{theorem}\label{theo:deczero}
  The protocol $\Sigma_{\text{DecZero}}$ in Table~\ref{tab:adaptedSigma} satisfies the three requirements of a $\Sigma$-protocol namely completeness, special soundness and special honest verifier zero knowledge.
\end{theorem}

\begin{proof}
  \textbf{Completeness:} For a prover that knows the underlying private key $s$ of the public key $k$, a verifier will always accept since: 
  \begin{align*}
    g^{z} &= g^{r'} \cdot (g^{s})^{e} = X \cdot k^{e}\\
    u^{z} &= u^{r'} \cdot (u^{s})^{e} = U \cdot v^{e}
  \end{align*}

  \textbf{Special soundness:} Suppose that $\left( X,U ; e_{1}; z_{1} \right)$ and $\left( X,U ; e_{2}; z_{2} \right)$ are two valid transcripts such that $e_{1} \neq e_{2}$.
  We have 
  \begin{equation*}
      \begin{cases}
        g^{z_1} &= X \cdot k^{e_1}\\
        g^{z_2} &= X \cdot k^{e_2}
      \end{cases}       
  \end{equation*}
  This implies that $g^{z_1-z_2} = k^{e_1-e_2}$ and $g^{\frac{z_1-z_2 }{e_1-e_2}} = k$.
  Since $s=log_{g}(k)$ then we succeed in extracting $s = \frac{z_1-z_2 }{e_1-e_2}$.

  \textbf{Special honest verifier Zero-Knowledge:}
  Consider $\mathcal{M}$ a simulator that is given an input $(u,v)$ and a challenge $e$.
  $\mathcal{M}$ operates as follows:
  \begin{enumerate}
    \item Chooses $z\in \mathbb{Z}_q $.
    \item Computes $X = g^{z} / k^{e}$ and $ U = u^{z} / v^{e}$.
    \item Outputs the transcript $\left( X,U ; e; z \right)$.
  \end{enumerate}
 Thus, $\mathcal{M}$ outputs a transcript of the same probability distribution as transcripts between the honest prover and verifier on common input $(u,v)$.

\end{proof}

\begin{theorem}\label{theo:blind}
  The protocol $\Sigma_{\text{Blind}}$ in Table~\ref{tab:adaptedSigma} satisfies the three requirements of a $\Sigma$-protocol namely completeness, special soundness and special honest verifier zero knowledge.
\end{theorem}

\begin{proof}
  \textbf{Completeness:} For a prover that knows the underlying blinding random $r$ used to produce $(a,b)=(u^r,v^r)$, a verifier will always accept since: 
  \begin{align*}
    u^{z} &= u^{r'} \cdot (u^{r})^{e} = U \cdot a^{e}\\
    v^{z} &= v^{r'} \cdot (v^{r})^{e} = V \cdot b^{e}
  \end{align*}

  \textbf{Special soundness:} Suppose that $\left( U,V ; e_{1}; z_{1} \right)$ and $\left( U,V ; e_{2}; z_{2} \right)$ are two valid transcripts such that $e_{1} \neq e_{2}$.
  We have 
  \begin{equation*}
      \begin{cases}
        u^{z_1} &= U \cdot a^{e_1}\\
        u^{z_2} &= U \cdot a^{e_2}
      \end{cases}       
  \end{equation*}
  This implies that $u^{z_1-z_2} = a^{e_1-e_2}$ and $u^{\frac{z_1-z_2 }{e_1-e_2}} = a$.
  Since $r=log_{u}(a)$ then we succeed in extracting $r = \frac{z_1-z_2 }{e_1-e_2}$.

  \textbf{Special honest verifier Zero-Knowledge:}
  Consider $\mathcal{M}$ a simulator that is given an input $\left((u,v), (a,b)\right)$ and a challenge $e$.
  $\mathcal{M}$ operates as follows:
  \begin{enumerate}
    \item Chooses $z\in \mathbb{Z}_q $.
    \item Computes $U = u^{z} / a^{e}$ and $ V = v^{z} / b^{e}$.
    \item Outputs the transcript $\left( U,V ; e; z \right)$.
  \end{enumerate}
 Thus, $\mathcal{M}$ outputs a transcript of the same probability distribution as transcripts between the honest prover and verifier on common input $\left((u,v), (a,b)\right)$.

\end{proof}

\begin{theorem}\label{theo:parial}
  The protocol $\Sigma_{\text{Partial}}$ in Table~\ref{tab:adaptedSigma} satisfies the three requirements of a $\Sigma$-protocol namely completeness, special soundness and special honest verifier zero knowledge.
\end{theorem}

\begin{proof}
  \textbf{Completeness:} For a prover that knows the underlying partial private key $sk_i$ of the joint key $pk_{joint}$ also the underlying private key of the public key $pk_i$, a verifier will always accept the proof on the input $\left((u,v), (u,c) \right)$ since: 
  \begin{align*}
    g^{z} &= g^{r'} \cdot (g^{sk_i})^{e} = X \cdot pk_i^{e}\\
    u^{z} &= u^{r'} \cdot (u^{sk_i})^{e} = U \cdot (v \cdot c^{-1})^{e} 
  \end{align*}
  This holds since $c = v \cdot u^{-sk_i}$.

  \textbf{Special soundness:} Suppose that $\left( X,U; e_{1}; z_{1} \right)$ and $\left( X,U ; e_{2}; z_{2} \right)$ are two valid transcripts such that $e_{1} \neq e_{2}$.
  We have 
  \begin{equation*}
      \begin{cases}
        g^{z_1} &= X \cdot pk_i^{e_1}\\
        g^{z_2} &= X \cdot pk_i^{e_2}
      \end{cases}       
  \end{equation*}
  This implies that $g^{z_1-z_2} = pk_i^{e_1-e_2}$ and $g^{\frac{z_1-z_2 }{e_1-e_2}} = pk_i$.
  Since $sk_i=log_{g}(pk_i)$ then we succeed in extracting $sk_i = \frac{z_1-z_2 }{e_1-e_2}$.

  \textbf{Special honest verifier Zero-Knowledge:}
  Consider $\mathcal{M}$ a simulator that is given an input $\left((u,v), (u,c), pk_i\right)$ and a challenge $e$.
  $\mathcal{M}$ operates as follows:
  \begin{enumerate}
    \item Chooses $z\in \mathbb{Z}_q $.
    \item Computes $X = g^{z} / pk_i^{e}$ and $ U = u^{z} / (v \cdot c^{-1})^{e}$.
    \item Outputs the transcript $\left( U,V ; e; z \right)$.
  \end{enumerate}
 Thus, $\mathcal{M}$ outputs a transcript of the same probability distribution as transcripts between the honest prover and verifier on common input $\left((u,v), (u,c), pk_i\right)$.

\end{proof}

\textbf{Adapted ZK-proofs:} 
In order to check and track the correctness of the computations over ElGamal encrypted data, we use three non-interactive ZK-proofs (see Table \ref{tab:adaptedZKs}) constructed from the three $\Sigma$-protocols (see Table \ref{tab:adaptedSigma}) using the Fiat-Shamir transformation where $\mathcal{H}: \mathbb{G}  \rightarrow   \left\{0,1\right\}^{t}$ is a hash function. 
In our construction, we use an AND proof of ZK-proofs, as shown in~\cite{hazay2010efficient}, using the same challenge for all individual ZK-proofs.

\begin{table*}[h!]  
    \centering
    \caption{Adapted Non-Interactive Zero-Knowledge proofs (NIZKs)}\label{tab:adaptedZKs}
    \begin{tabular}{l|l|l|l}
      \toprule 
      & $\text{NIZK}_{\text{DecZero}}$
      & $\text{NIZK}_{\text{Blind}}$
      & $\text{NIZK}_{\text{Partial}}$\\ 
      & for $ [m_1] - [m_2] = (u,v)$
      & for $ [m]^r= (u^r,v^r) = (a,b)$
      & for $[m]_i=(u, v\cdot u^{-sk_i})= (u,c)$\\ 
      \midrule 
      Commitment& $X = g^{r'}, U = u^{r'}$ & $U = u^{r'}, V = v^{r'}$&  $X = g^{r'}, U = u^{r'}$ \\
      \midrule 
      Challenge& $e = \mathcal{H} \left( k, X,U  \right)$& $e = \mathcal{H} \left( k, U, V  \right)$ & $e = \mathcal{H} \left( pk_i, X, U \right)$\\
      \midrule
      Response&  $z= r'+e \cdot s $& $z= r'+e \cdot r $&  $z= r'+e \cdot sk_i $\\
      \midrule
      Verification Equations& $ e =^{?} \mathcal{H} \left( k, X,U  \right)$ & $ e =^{?} \mathcal{H} \left( k, U, V  \right)$& $e = \mathcal{H} \left( pk_i, X, U \right)$\\
      & $g^{z} =^{?} X \cdot k^{e}$ & $u^{z} =^{?} U \cdot a^{e}$ & $g^{z} =^{?} X \cdot pk_{i}^{e}$\\
      & $u^{z} =^{?} U \cdot v^{e}$ & $v^{z} =^{?} V \cdot b^{e}$ & $u^{z} =^{?} U \cdot (v \cdot c^{-1})^{e}$\\
      \bottomrule 
    \end{tabular}
  
\end{table*}

\subsubsection{Verification Protocol Secure Against Malicious Adversaries}

For the transformation, we introduce a trusted entity, which is merely involved during the enrollment, called \emph{enrollment server} that holds a signature key pair $(v_{enr},s_{enr})$. 
Unlike the enrollment in our protocol Figure~\ref{fig:shverification}, the client has an additional key pair $(k_{clt},s_{clt})$ for an additively homomorphic ElGamal encryption and generates $k$ pseudo-random permutations (PRP) $ \left(\pi_{1}, \cdots, \pi_{k} \right)$. 
The template is formed differently as well.
After selecting the corresponding rows as in Equation \eqref{eq:templateSH}, the client performs the following modifications on Equation \eqref{eq:encryptedtemplateSH}.
The component $\ct{s_{f_i,j}^{i}}$ becomes a vector that contains the encryption, under $k_{clt}$, of its column position (i.e. $\left[ j \right]$) and an index $r_{i,j} = \pi_{i}(j)$ by which this component will be located later.
The components of the i$^{\text{th}}$ vector are ordered according to the indexes $r_{i,j}$.
Thus $\ct{\mat{T'_{u\text{ID}}}}$ becomes
\begin{equation}\label{eq:encryptedtemplateMAL}
  \ct{\mat{T'_{u\text{ID}}}} = \left(  \left( r_{i,j}, \left[ j \right], \ct{s_{f_i,j}^{i}} \right)_{j\in [1,n]} \right)_{i\in [1,k]}
\end{equation}
The client then sends $\ct{\mat{T'_{u\text{ID}}}}$ to the enrollment server who appends to each component two signatures: $\sigma_{i,j} = \text{Sign} \left( s_{enr}, r_{i,j}, [j], u\text{ID}  \right)$ that binds the index $r_{i,j}$ with encryption of the column position $j$ and $u\text{ID}$; and $\alpha_{i,j} = \text{Sign} \left( s_{enr}, [j], \ct{ s_{f_i,j}^{i} }, u\text{ID} \right)$ that binds $[j]$ with $\ct{ s_{f_i,j}^{i} }$ and $u\text{ID}$.
Those signatures ensure the authenticity of the template during the verification.
The final protected template has the form:

\begin{equation}\label{eq:encryptedtemplateMAL}
  \ct{\mat{T_{u\text{ID}}}} = \left(  \left( r_{i,j}, \left[ j \right], \ct{s_{f_i,j}^{i}},\sigma_{i,j},\alpha_{i,j} \right)_{j\in [1,n]} \right)_{i\in [1,k]}
\end{equation}

Besides, the enrollment server generates a permutation $\pi_{thr}$ to form $\mat{\Theta}$ the \emph{permuted-encrypted threshold vector} under $pk_{joint}$.
\begin{equation}\label{eq:thresholdvect}
  \mat{\Theta} = \left( \ct{ \pi_{thr}(i) } \right)_{i\in [\theta,S_{\max}]}
\end{equation}
Note that only the enrollment server knows the order of the plain values of $\mat{\Theta}$.
Finally, it sends $\mat{\Theta}$ to the client and $u\text{ID}$, $\ct{\mat{T_{u\text{ID}}}}$  and $\mat{\Theta}$ to the server that stores them for later retrieval.
This template's structure facilitates the transition from a semi-honest construction to one that is more suitable to withstand malicious adversaries.

Figure~\ref{fig:mal-biometric-verification} describes our biometric verification protocol secure against malicious adversaries.
Unlike in our protocol Figure~\ref{fig:shverification}, in step \circled{2}, the client sends its probe encrypted $\left[ \mat{P} \right]$ and proves the knowledge of the underlying plain probe. 
He also sends the corresponding indexes $\mat{R}$ to allow the server to locate the desired components.
The server, in step \circled{3}, sends the first half of the components to allow the client to prove that the requested components correspond to the ones that appear in the protected template.
Once the server is convinced, he then sends, in step \circled{4}, the second half that contains the encrypted scores.
Before engaging in any proof, the client checks the authenticity of the received data by verifying the signatures $\sigma_i$ and $\alpha_i$, in steps \circled{3} and \circled{4}. 
In step \circled{4}, with respect to the same order determined by $\mat{R}$ as well as $\mat{\Theta}$, both the client and the server compute $\ct{S}$ and $\ct{\mat{C}}$.
Note that here they must have the same resulted values with the same randomness.
Next only the server blinds the comparison vector $\ct{\mat{C}}$ using a vector of random values $\mat{a}$ to get $\ct{\mat{aC}}$ and partially decrypts it as $\left[\mat{ aC }\right]_{ser}$.
He then proves to the client that $\ct{\mat{aC}}$ is a blind version of $\ct{\mat{C}}$ and that $\left[\mat{ aC }\right]_{ser}$ is its partial decryption. 
If the client is convinced, he runs the final decryption and parses the blinded comparison vector $\mat{aC}$.
If there is a zero, he outputs a \emph{match} otherwise outputs a \emph{no match}.  

\begin{figure*}[!h]
  \centering
  \includegraphics[width=\textwidth]{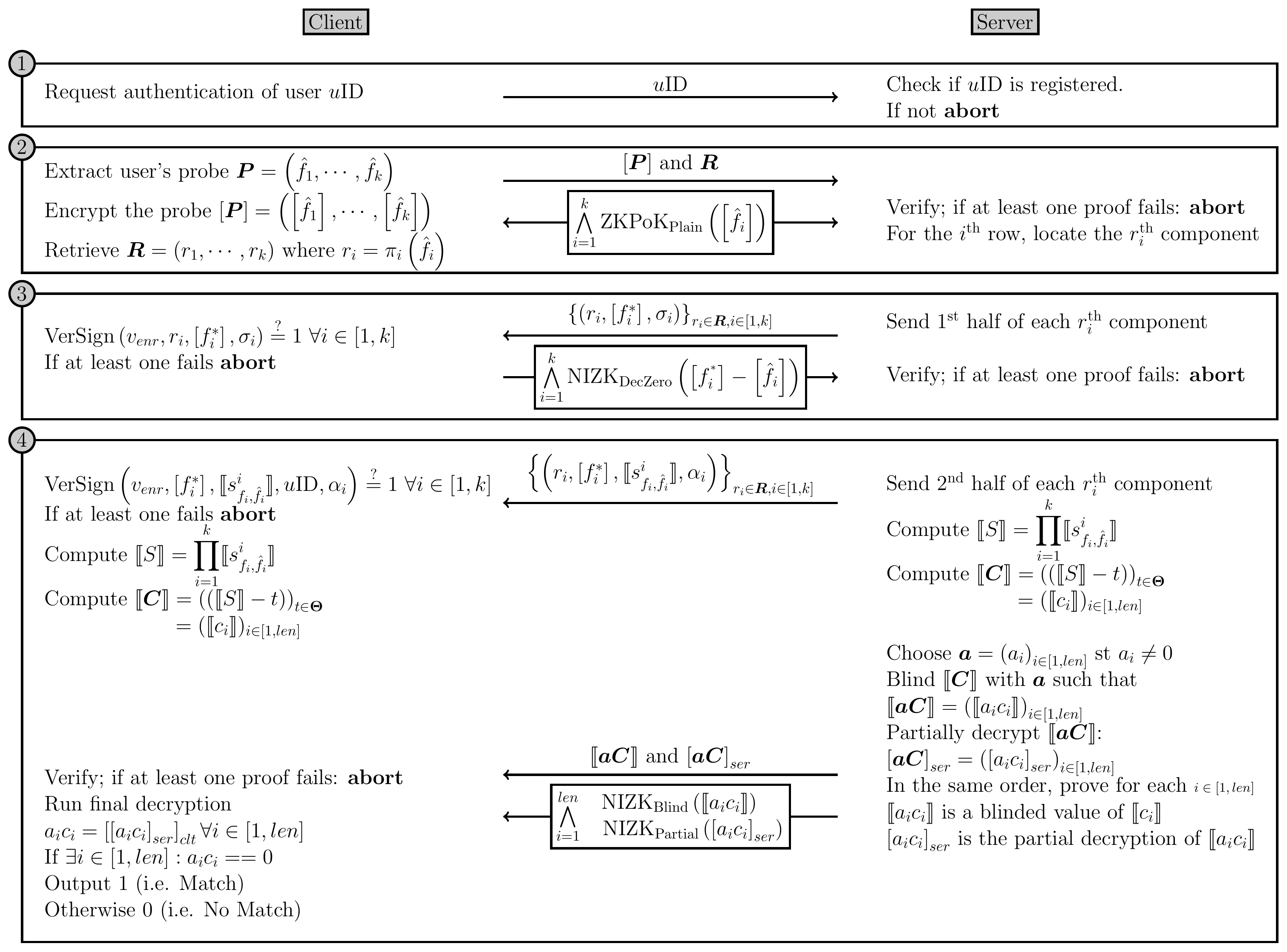}
  \caption{Biometric Verification Protocol Secure Against Malicious Adversaries}
  \label{fig:mal-biometric-verification}
\end{figure*}

\section{Security Analysis}
As mentioned in Section~\ref{section:protocols}, the enrollment is performed in a fully controlled environment.
Therefore, we only discuss the security of verification protocols by separately analyzing the case of compromised client and the case of compromised server.
Regardless noise measurement, the correctness of both protocols is straightforward.
For a perfectly captured user's probe, an honest client who is interacting with an honest server will yield a \emph{match} if the user is genuine and a \emph{no match} if the user is an impostor.

\subsection{Security Proof of Protocol~\ref{fig:shverification}}
\textbf{Semi-honestly compromised client:} he receives first the template (\ref{eq:encryptedtemplateSH}) that is encrypted under threshold ElGamal; which means that the client can not decrypt on his own to learn the scores.
However, he may tend to use the public HELR lookup tables together with the encrypted template and to learn which rows were encrypted.
Thanks to the IND-CPA property of ElGamal, the client is unable to link the i$^{\text{th}}$ vector of $\ct{\mat{T_{u\text{ID}}}}$ to any row of the i$^{\text{th}}$ public HELR lookup table. 
In the second round he receives the partial decryption of the comparison vector that was blinded and permuted by the server.
Again the client fails in attempting to infer any significant information from the comparison vector.
Given that $\mat{aC}$ was blinded and permuted by the server, the only leaked information here is the comparison outcome. 
In case of a match the client will find a value zero in a random position of the comparison vector.
In case of a no match, all values he will obtain are random; thus in both cases he can not learn the final score $S$.
Moreover, the client may try to combine both the encrypted $\ct{\mat{T_{u\text{ID}}}}$ and the blinded-permuted comparison vector $\mat{aC}$; again nothing can be inferred thanks to the robustness of threshold ElGamal, the applied permutation and the randomness introduced by the blinding.

\textbf{Semi-honestly compromised server:} regardless the user's identity, it receives only one message from the client that is the threshold encrypted final score $\ct{S}$; so the server can not decrypt on his own.
Given the fact that $\ct{S}$ is only a multiplication of the chosen, accordingly to the probe, encrypted individual scores from the template, the server may try to select some encrypted values from the template and compare their multiplication with the received $\ct{S}$. 
This comparison is meaningless since the client re-randomizes $\ct{S}$ before sending it and thanks to the IND-CPA property of ElGamal the server is unable to learn any information about the probe.

Both cases demonstrate that no biometric information is leaked except the comparison result; which is the output of the protocol. 
Therefore, our protocol in Figure~\ref{fig:shverification} is secure in the semi-honest model.

\subsection{Security Proof of Protocol~\ref{fig:mal-biometric-verification}}

\begin{theorem}\label{theo:mal}
  Assume that $\mathcal{H}$ is a collision resistant hash function and the signature scheme is EUF-CMA secure then the protocol in Figure~\ref{fig:mal-biometric-verification} is secure in the ( $\mathpzc{f}_{\text{ZKPoK}}^{R_{\text{Plain}}} $, $\mathpzc{f}_{\text{NIZK}}^{R_{\text{DecZero}}}$, $\mathpzc{f}_{\text{NIZK}}^{R_{\text{Blind}}}$, $\mathpzc{f}_{\text{NIZK}}^{R_{\text{Partial}}}$)-Hybrid world in the presence of malicious adversary.
\end{theorem}

\begin{proof}
\textbf{Maliciously compromised client:} 
Let $\mathcal{A}$ be an adversary operating in the ($\mathpzc{f}_{\text{ZKPoK}}^{R_{\text{Plain}}} $, $\mathpzc{f}_{\text{NIZK}}^{R_{\text{DecZero}}}$, $\mathpzc{f}_{\text{NIZK}}^{R_{\text{Blind}}}$, $\mathpzc{f}_{\text{NIZK}}^{R_{\text{Partial}}}$)-Hybrid world that is controlling the client, we construct $\mathcal{S}_{\mathcal{A}}$ an adversary operating in the ideal world that uses $\mathcal{A}$ as a subroutine. 
$\mathcal{S}_{\mathcal{A}}$ is given the description of $\mathcal{A}$, the probe $\mat{P}$, the user's identity $u\text{ID}$, permuted-encrypted threshold vector $\mat{\Theta}$ \eqref{eq:thresholdvect}, $v_{enr}$, $k_{clt}$, $pk_{clt}$, $pk_{joint}$, $sk_{clt}$ and $k_{clt}$. 
Moreover, $\mathcal{S}_{\mathcal{A}}$ is also given the public input of the server, i.e. the set of templates $\left\{\ct{\mat{T_{u\text{ID}}}}\right\}_{u\text{ID}}$ \eqref{eq:encryptedtemplateMAL} and $pk_{ser}$.

During the simulation, $\mathcal{S}_{\mathcal{A}}$ will be playing the role of the honest server and the trusted party that computes the functionalities used in the hybrid world.
We describe $\mathcal{S}_{\mathcal{A}}$ as follows:

\circled{1} $\mathcal{S}_{\mathcal{A}}$ receives $u\text{ID}$ from $\mathcal{A}$ and verifies if $u\text{ID}$ has been enrolled before if not \textbf{abort}.

\circled{2} $\mathcal{S}_{\mathcal{A}}$ receives $[\mat{P_{\mathcal{A}}}]$ and $\mat{R_{\mathcal{A}}}$ from $\mathcal{A}$ and a request containing the statements $[\mat{P_{\mathcal{A}}}]$ and the witness $\mat{P_{\mathcal{A}}} = \left( \widehat{f_{i,\mathcal{A}}} \right)_{i \in [1,k]}$ to be sent to $\mathpzc{f}_{\text{ZKPoK}}^{R_{\text{Plain}}} $.
At this stage, $\mathcal{S}_{\mathcal{A}}$ sends $u\text{ID}$ and $\mat{P_{\mathcal{A}}}$ to the trusted party and receives back $b \in \{0,1\}$ the public output of the client.
Here $0$ means \emph{no match} and $1$ means \emph{match}.

$\mathcal{S}_{\mathcal{A}}$ verifies whether the statements are consistent with the witness if not \textbf{abort}.
According to $\mat{R_{\mathcal{A}}}$ and user's $u\text{ID}$ template, $\mathcal{S}_{\mathcal{A}}$ sends to $\mathcal{A}$ the first half of the requested components extracted from the set of templates $\left\{\ct{\mat{T_{u\text{ID}}}}\right\}_{u\text{ID}}$.

\circled{3} $\mathcal{S}_{\mathcal{A}}$ receives the statements $[\mat{P_{\mathcal{A}}}]$ and $\left( \left[ f^{*}_{i} \right] \right)_{i\in [1,k]}$ and witnesses that $\mathcal{A}$ sends to $\mathpzc{f}_{\text{NIZK}}^{R_{\text{DecZero}}}$.
$\mathcal{S}_{\mathcal{A}}$ verifies whether the statements are consistent with the witness if not \textbf{abort}.
$\mathcal{S}_{\mathcal{A}}$ sends the second half of the requested components to $\mathcal{A}$.

\circled{4} At this stage there are two cases to distinguish depending on
the output $b \in \{0, 1\}$

\textbf{Case $b=0$:} $\mathcal{S}_{\mathcal{A}}$ generates a random vector $ \mat{aC_{\mathcal{S}_{\mathcal{A}}}} =\left( t_i \right)_{i\in [1,len]}$ such that $\forall i \in [1, len] \ \ t_i \neq 0$ then encrypts it first using $pk_{joint}$ to get $\ct{ \mat{aC_{\mathcal{S}_{\mathcal{A}}}}} $ then encrypts it using $pk_{clt}$ to get $ \left[ \mat{aC_{\mathcal{S}_{\mathcal{A}}}} \right]_{ser}$ so that $\mathcal{A}$ will succeed in decrypting it and the resulted vector will contain non-zero values. 
This is possible since in ElGamal a partial decryption under $sk_{ser}$ yields an encryption under $pk_{clt}$ and the final decryption of an ElGamal threshold encryption is the decryption of an ElGamal non-threshold encryption.

\textbf{Case $b=1$:} $\mathcal{S}_{\mathcal{A}}$ generates a random vector $ \mat{aC_{\mathcal{S}_{\mathcal{A}}}} =\left( t_i \right)_{i\in [1,len]}$ such that $t_j = 0$ and $\forall i \neq j : \ t_i \neq 0$ then encrypts it first using $pk_{joint}$ to get $\ct{ \mat{aC_{\mathcal{S}_{\mathcal{A}}}}} $ then encrypts it using $pk_{clt}$ to get $ \left[ \mat{aC_{\mathcal{S}_{\mathcal{A}}}} \right]_{ser}$ so that $\mathcal{A}$ will succeed in decrypting it and the resulted vector will contain exactly one zero.
This is also possible for the same reasons mentioned in case $b=0$.

Since the simulator can not partially decrypt on behalf of the server, the only alternative for him is to encrypt using $pk_{clt}$.	
Then sends both vectors $\ct{ \mat{aC_{\mathcal{S}_{\mathcal{A}}}}} $ and $ \left[ \mat{aC_{\mathcal{S}_{\mathcal{A}}}} \right]_{ser}$ to $\mathcal{A}$ who requests the answer of $\mathpzc{f}_{\text{NIZK}}^{R_{\text{Blind}}}$ and $\mathpzc{f}_{\text{NIZK}}^{R_{\text{Partial}}}$ on these statements.
$\mathcal{S}_{\mathcal{A}}$ hands 1 to $\mathcal{A}$ and outputs whatever $\mathcal{A}$ does.

Given the fact that $\mathcal{S}_{\mathcal{A}}$ outputs whatever $\mathcal{A}$ does, what remains to be proven is the indistinguishability of $\mathcal{A}$'s view in both worlds: ($\mathpzc{f}_{\text{ZKPoK}}^{R_{\text{Plain}}} $, $\mathpzc{f}_{\text{NIZK}}^{R_{\text{DecZero}}}$, $\mathpzc{f}_{\text{NIZK}}^{R_{\text{Blind}}}$, $\mathpzc{f}_{\text{NIZK}}^{R_{\text{Partial}}}$)-Hybrid and ideal (simulation).
$\mathcal{A}$'s view consists of sent-and-received  messages. 
Here the difference lies in the fact that the received messages were generated by the honest server in the real execution while, during the simulation, they were generated by the simulator $\mathcal{S}_{\mathcal{A}}$.

The steps \circled{1}, \circled{2} and \circled{3} are identical to the real execution since the only thing $\mathcal{S}_{\mathcal{A}}$ sends is some components from the template.

In step \circled{4}, in the real execution, the honest server always forms valid proofs. 
However in the simulation, $\mathcal{S}_{\mathcal{A}}$ hands $\mathcal{A}$ fake valid proofs according to the output that it has received from the trusted party.
As a result, the indistinguishability is maintained since $\mathcal{A}$ in both worlds receives the blinded vector encrypted using $pk_{joint}$ (IND-CPA property of the threshold ElGamal) and a ciphertext decryptable under its own private key $sk_{clt}$ (IND-CPA property of non-threshold ElGamal) and 1 as an answer from $\mathpzc{f}_{\text{NIZK}}^{R_{\text{Blind}}}$ and $\mathpzc{f}_{\text{NIZK}}^{R_{\text{Partial}}}$.

Therefore, ($\mathpzc{f}_{\text{ZKPoK}}^{R_{\text{Plain}}} $, $\mathpzc{f}_{\text{NIZK}}^{R_{\text{DecZero}}}$, $\mathpzc{f}_{\text{NIZK}}^{R_{\text{Blind}}}$, $\mathpzc{f}_{\text{NIZK}}^{R_{\text{Partial}}}$)-Hybrid and simulation are indistinguishable.

\textbf{Maliciously compromised server:}
Let $\mathcal{A}$ be an adversary operating in the ($\mathpzc{f}_{\text{ZKPoK}}^{R_{\text{Plain}}} $, $\mathpzc{f}_{\text{NIZK}}^{R_{\text{DecZero}}}$, $\mathpzc{f}_{\text{NIZK}}^{R_{\text{Blind}}}$, $\mathpzc{f}_{\text{NIZK}}^{R_{\text{Partial}}}$)-Hybrid world controlling the server, we construct $\mathcal{S}_{\mathcal{A}}$ an adversary operating in the ideal world using $\mathcal{A}$ as a subroutine. 
$\mathcal{S}_{\mathcal{A}}$ is given the description of $\mathcal{A}$, set of templates $\left\{\ct{\mat{T_{u\text{ID}}}}\right\}_{u\text{ID}}$ (\ref{eq:encryptedtemplateMAL}), permuted-encrypted threshold vector $\mat{\Theta}$ (\ref{eq:thresholdvect}), $v_{enr}$, $pk_{ser}$, $pk_{joint}$ and $sk_{ser}$. 
Moreover, $\mathcal{S}_{\mathcal{A}}$ is also given the public input of client $u\text{ID}$ the user's identity, $k_{clt}$ and $pk_{clt}$.

During the simulation, $\mathcal{S}_{\mathcal{A}}$ will be playing the role of the honest client and the trusted party that computes the functionalities used in the hybrid world.
We describe $\mathcal{S}_{\mathcal{A}}$ as follows:

\circled{1} $\mathcal{S}_{\mathcal{A}}$ verifies if $u\text{ID}$ has been enrolled before if not \textbf{abort}.	

\circled{2} $\mathcal{S}_{\mathcal{A}}$ generates a random probe $ \mat{P_{\mathcal{S}_{\mathcal{A}}}}=\left( \hat{f^*_i} \right)_{i\in [1,k]}$ as a $k-$dimensional feature vector and encrypts it using $k_{clt}$ to get $ \left[ \mat{P_{\mathcal{S}_{\mathcal{A}}}} \right]$.
$\mathcal{S}_{\mathcal{A}}$ generates another random vector $\mat{ R_{\mathcal{S}_{\mathcal{A}}} } = (r^*_i)_{i\in [1,k]}$ where $r_i \in [1, n]$. 
Then sends $u\text{ID}$, $ \left[ \mat{P_{\mathcal{S}_{\mathcal{A}}}} \right]$ and $\mat{ R_{\mathcal{S}_{\mathcal{A}}} }$ to $\mathcal{A}$.
$\mathcal{S}_{\mathcal{A}}$ receives the statement $[\mat{P_{\mathcal{S}_{\mathcal{A}}}}]$ that $\mathcal{A}$ sent to $\mathpzc{f}_{\text{ZKPoK}}^{R_{\text{Plain}}} $ and hands 1 to $\mathcal{A}$.

\circled{3} $\mathcal{S}_{\mathcal{A}}$ receives from $\mathcal{A}$ the first half of the requested components and verifies the signatures $\sigma_i$. 
If at least one of them fails \textbf{abort} otherwise \textbf{continue}. 
$\mathcal{S}_{\mathcal{A}}$ receives the statements $[\mat{P_{\mathcal{S}_{\mathcal{A}}}}]$ and $([f^{*}_{i}])_{i\in [1,k]}$ that $\mathcal{A}$ sent to $\mathpzc{f}_{\text{NIZK}}^{R_{\text{DecZero}}}$ then hands 1 to $\mathcal{A}$.

\circled{4} $\mathcal{S}_{\mathcal{A}}$ receives the second half of the requested components and verifies the signatures $\alpha_i$. If at least one of them fails \textbf{abort} otherwise \textbf{continue}.
$\mathcal{S}_{\mathcal{A}}$ receives the statements $\llbracket S \rrbracket$, $ \bigl \llbracket \mat{C} \bigr \rrbracket $, $ \bigl \llbracket\mat{ aC } \bigr \rrbracket $, $ \left[ \mat{ aC } \right]_{ser} $ along with their proper witnesses that $\mathcal{A}$ sent to $\mathpzc{f}_{\text{NIZK}}^{R_{\text{Blind}}}$ and $\mathpzc{f}_{\text{NIZK}}^{R_{\text{Partial}}}$.
$\mathcal{S}_{\mathcal{A}}$ checks the statements and their witnesses, if at least one of them is not correct \textbf{abort} otherwise $\mathcal{S}_{\mathcal{A}}$ outputs whatever $\mathcal{A}$ does.

Since $\mathcal{S}_{\mathcal{A}}$ outputs whatever $\mathcal{A}$ does, what remains to be proven is the indistinguishability of $\mathcal{A}$'s view in both worlds: ($\mathpzc{f}_{\text{ZKPoK}}^{R_{\text{Plain}}} $, $\mathpzc{f}_{\text{NIZK}}^{R_{\text{DecZero}}}$, $\mathpzc{f}_{\text{NIZK}}^{R_{\text{Blind}}}$, $\mathpzc{f}_{\text{NIZK}}^{R_{\text{Partial}}}$)-Hybrid and ideal (simulation).
$\mathcal{A}$'s view consists of sent-and-received messages. 
Here the difference lies in the fact that the received messages were generated by the honest client in the real execution while, during the simulation, they were generated by the simulator $\mathcal{S}_{\mathcal{A}}$.

In both executions, the steps \circled{1} and \circled{2} are indistinguishable.
Indeed, the IND-CPA property of ElGamal ensures the indistinguishability between the random encrypted probe $[\mat{P_{\mathcal{S}_{\mathcal{A}}}}]$ and the real encrypted probe $[\mat{P}]$.
The permutation used by the honest client is a PRP permutation. 
So in both worlds, $\mathcal{A}$ is unable to distinguish between $(r_i)_{i\in [1,k]}$ received from the honest client and $(r^*_i)_{i\in [1,k]}$ received from $\mathcal{S}_{\mathcal{A}}$.
$\mathcal{S}_{\mathcal{A}}$ has generated then encrypted the vector $\mat{P_{\mathcal{S}_{\mathcal{A}}}}$ correctly (as the honest client would do). 
As a consequence, it gives $\mathcal{A}$ the value 1 instead of querying the ideal functionality $\mathpzc{f}_{\text{ZKPoK}}^{R_{\text{Plain}}} $.

In step \circled{3}, $\mathcal{S}_{\mathcal{A}}$ checks the signatures $\sigma_i$ and if at least one of them fails he aborts as the honest client would do. 
Then $\mathcal{A}$ always receives 1 from $\mathcal{S}_{\mathcal{A}}$ as an answer for $\mathpzc{f}_{\text{NIZK}}^{R_{\text{DecZero}}}$ which is always the case in the hybrid execution where the client is assumed to be honest. 
As a result, step \circled{3} is indistinguishable in both executions.

In step \circled{4}, $\mathcal{S}_{\mathcal{A}}$ checks the signatures $\alpha_i$ and if the honest client would abort $\mathcal{S}_{\mathcal{A}}$ will also abort.
Then it checks the statements and the corresponding witnesses for $\mathpzc{f}_{\text{NIZK}}^{R_{\text{Blind}}}$ and $\mathpzc{f}_{\text{NIZK}}^{R_{\text{Partial}}}$. 
If at least one of them is incorrect he aborts as the honest client would do. 
If not it outputs whatever $\mathcal{A}$ does.

Therefore, ($\mathpzc{f}_{\text{ZKPoK}}^{R_{\text{Plain}}} $, $\mathpzc{f}_{\text{NIZK}}^{R_{\text{DecZero}}}$, $\mathpzc{f}_{\text{NIZK}}^{R_{\text{Blind}}}$, $\mathpzc{f}_{\text{NIZK}}^{R_{\text{Partial}}}$)-Hybrid and simulation are indistinguishable.

The cases client malicious and server malicious conclude our proof. 
\end{proof}

\section{Experimentation and Evaluation}
We used a $64$-bit computer Intel(R) Core $i7$-$8650$U CPU with $4$ cores ($8$ logical processors) rated at $2.11$GHz and $16$GB of memory. 
We used Python~$3.5.2$ to implement the HELR classifier and generated the lookup tables on the cluster.
For the protocols, we used Linux Ubuntu $18.04.4$ LTS run on Windows-SubLinux (WSL) on Windows $10$.

\subsection{Generation of HELR lookup tables}
Our approach supports any biometric modality that can be encoded as a fixed-length real-valued feature vector.
To show this we conducted experiments on \emph{dynamic signatures} (BMDB~\cite{ortega2009multiscenario} ) and \emph{faces} (PUT~\cite{kasinski2008put} and FRGC2.0 (Experiment 1, mask II)~\cite{phillips2005overview} datasets).

For the BMDB dataset, the initial features were extracted by the algorithm described in \cite{gomez2016implementation} and were rendered i.i.d. using PCA and LDA as described in \cite{bazen2004likelihood}.
Their dimensionality was reduced to $36$ and then they were split into a training set containing $1350$ feature vectors ($45$ users) and a test set containing $4800$ feature vectors ($160$ different users).
For the PUT dataset, the initial features were extracted by the VGG16 network \cite{parkhi2015deep} (the first layers' weights were taken from \cite{vggfaceweights} and the last layer was retrained for PUT) and were rendered i.i.d. as above.
Their dimensionality was reduced to $49$ and then they were split into a training set containing $1100$ feature vectors ($50$ users) and a test set containing $1095$ feature vectors ($50$ different users).
For the FRGC dataset, the initial features were extracted by the VGG-Face network \cite{vggfacedescriptor} and were rendered i.i.d. as above.
Their dimensionality was reduced to $94$ and then they were split into a training set containing $12776$ feature vectors ($222$ users) and a test set containing $16028$ feature vectors ($466$ users) with identity overlap between both sets of $153$ users.
Because of the identity overlap, we use FRGC only for evaluating the speed of our protocols on different HELR parameters.

We generated the HELR lookup tables using the training set and measured the HELR classifier performance using the test set.
Figure~\ref{fig:detcurves} depicts the DET curves of LLR and HELR on the three datasets generated based on the parameters in Table~\ref{tab:helrparams}. 
For those parameters, the HELR performs better than the LLR.
This can be justified by the data-driven nature of LLR which makes it prone to slightly overfitting on the data.
HELR addresses this model mismatch through applying quantization and rounding.
Plus, HELR achieves an EER of $2.4\%$ for BMDB outperforming \cite{gomez2016implementation} and an EER of $0.27\%$ for PUT and $0.25\%$ for FRGC that is similar to the state-of-the-art $0.2\%$~\cite{deng2019arcface}.

\begin{table*}[h!]  
    \centering    
    \caption{Parameters of HELR lookup tables generated from BMDB, PUT and FRGC datasets}    
    \label{tab:helrparams}
    \begin{tabular}{c|c|c|c|c|c|c|c|c|c|c}
      \toprule 
      \multirow{2}{*}{Dataset} &\multirow{2}{*}{ Modality}  & \multirow{2}{*}{\makecell{\#Features\\after LDA}} &\multirow{2}{*}{\makecell{HELR\\dimension}}  & \multirow{2}{*}{\makecell{Score step\\$\Delta $}}   
      &  \multirow{2}{*}{\makecell{Threshold $\theta$\\at $0.1 \% $ FMR }} & \multirow{2}{*}{$S_{\max}$  }     &  \multirow{2}{*}{\makecell{\#Genuine\\comparison}} & \multirow{2}{*}{\makecell{\#Impostor\\comparison}} & \multirow{2}{*}{EER }
      &  \multirow{2}{*}{\makecell{FNMR\\ at $0.1 \% $ FMR }} \\
      &&&&& &&&&&\\
      \midrule 
      BMDB & Signature & $36$ & $16 \times 16$ & $0.5$ & $14$ & $99$ & $6.96 \times 10^4$ 
      & $1.14 \times 10^7$  &  $2.4 \% $ &  $11.07 \% $ \\
      \midrule 
      PUT & Face & $49$ & $64 \times 64$ & $1$ & $-53$ & $82$ & $1.14 \times 10^4$ 
      & $5.87 \times 10^5$  &  $0.27 \% $ &  $0.81 \% $ \\
      \midrule 
      FRGC & Face & $94$ & $64 \times 64$ & $1.5$ & $-1$ & $73$ & $3.99 \times 10^5 $ 
      &  $1.28 \times 10^8$  &  $0.25 \% $ &  $0.49 \% $ \\
      \bottomrule 
    \end{tabular}  
\end{table*}

\begin{figure}[!h]
  \centering
  \includegraphics[width=\columnwidth]{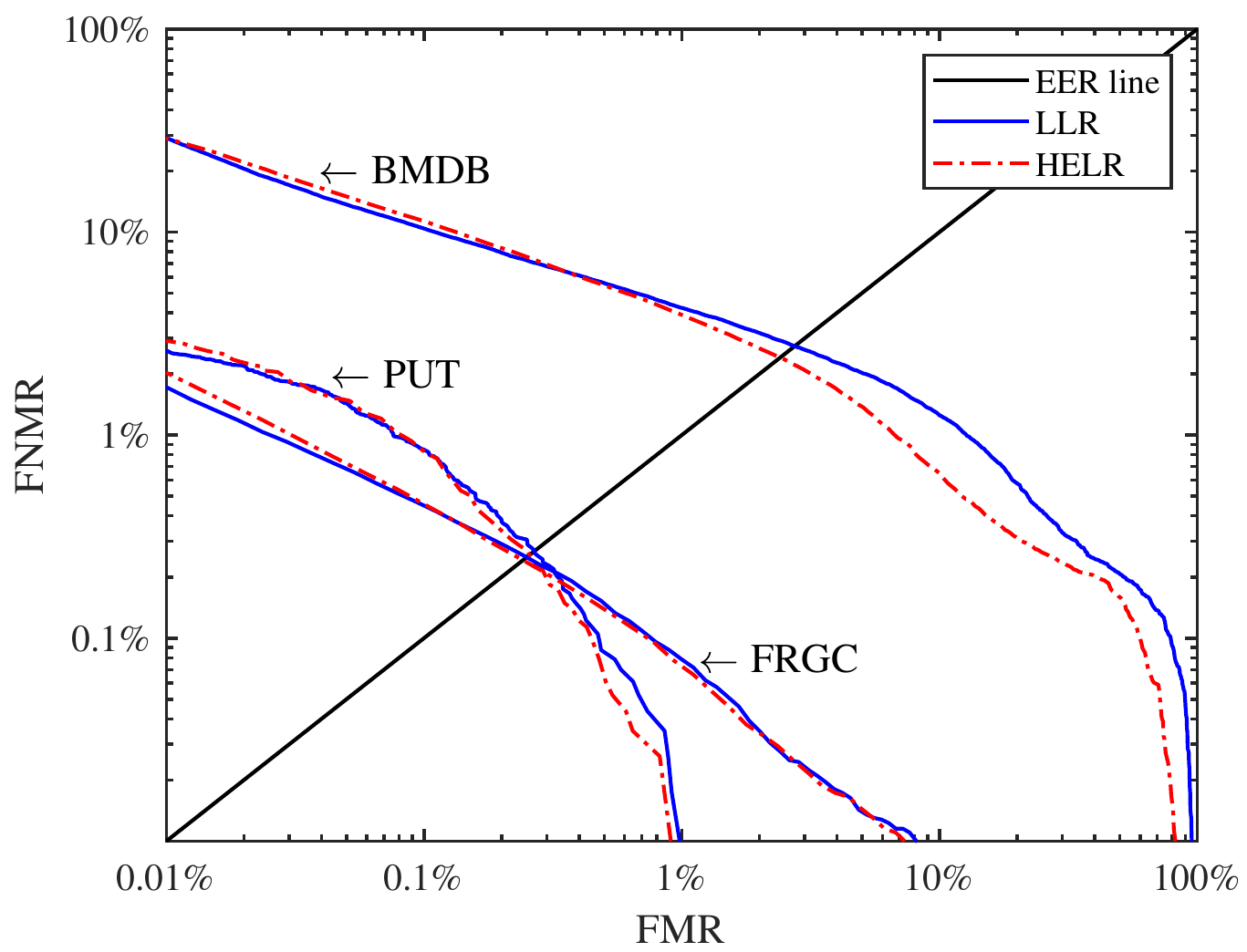}
  \caption{DET curves of BMDB, PUT and FRGC datasets}\label{fig:detcurves}
\end{figure}

\subsection{Verification Protocols}
We implemented the prototype of both protocols as described in Figure~\ref{fig:shverification} and Figure~\ref{fig:mal-biometric-verification} following the client/server architecture in C++ using libscapi library~\cite{bar2016libscapi} for secure two-party computation and OpenMP~\cite{dagum1998openmp} to parallelize the implementation.

Figure~\ref{fig:verifSHandMALruntime} shows the runtime, per elliptic curve, of $500$ genuine verifications via protocols \ref{fig:shverification} and \ref{fig:mal-biometric-verification} tested on the three datasets and their median runtime is given in Table~\ref{tab:summaryTable}.
The runtime increases as the size of the elliptic curve's prime increases implying the increase of the security strength.
According to \cite{bluekrypt}, the security strength of P$192$, P$224$ and P$256$ corresponds respectively to $96$, $112$ and $128$ bits.
We recall that in the last step of our protocols, the client decrypts a permuted vector and makes his comparison decision as soon as the first zero is found.
The server does not learn the comparison outcome as he is not contacted again after he sends this vector.  
The decryption of a permuted vector justifies the boxplots overall outliers in Figure~\ref{fig:verifSHandMALruntime}. 
The runtime of a false non-match (genuine) or a true non-match (impostor) is slower, approximately the upper outliers of the boxplots in Figure~\ref{fig:verifSHandMALruntime}, since the vector will be fully parsed searching for a zero. 
Among HELR parameters (see Table \ref{tab:helrparams}), the comparison vector length ($S_{\max} - \theta + 1$) has more influence on the runtime since the test on FRGC shows the fastest runtime although it has the largest parameters except for the vector length.
This can be justified by the expensiveness of the decryption operation. 
\begin{figure*}[!h]
  \centering
  \includegraphics[width=\linewidth]{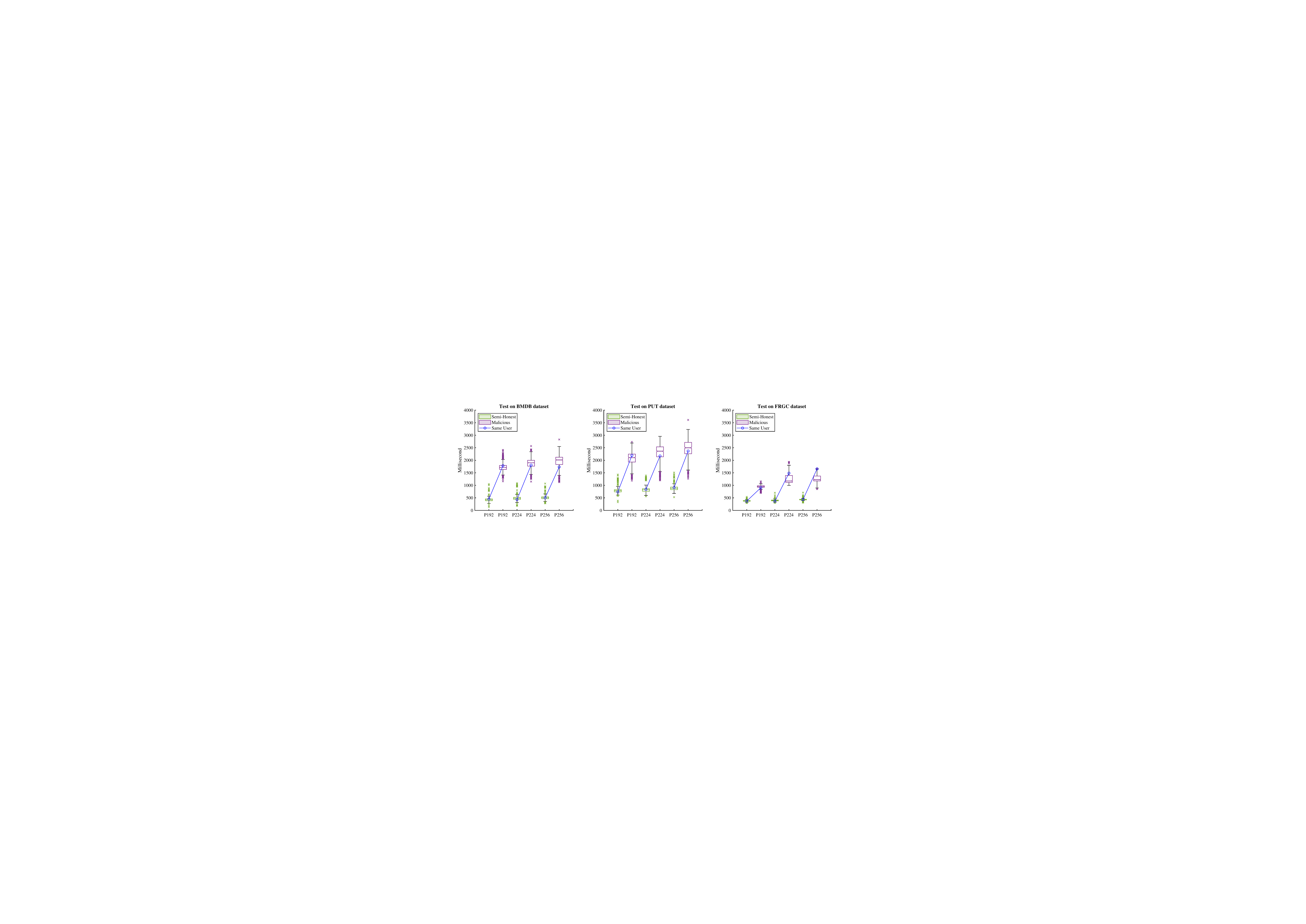}
  \caption{Runtime of Protocol \ref{fig:shverification} (green) and Protocol \ref{fig:mal-biometric-verification} (purple) on $500$ genuine verifications over three elliptic curves.} \label{fig:verifSHandMALruntime}
\end{figure*}

Table~\ref{tab:tempSize} shows that the template size in Protocol \ref{fig:shverification} is smaller than the one in Protocol \ref{fig:mal-biometric-verification} since it depends on the feature level, the number of features, the elliptic curve's prime and the adversarial model.
The template size could be decreased by optimizing its storage.  
\begin{table}[h!]
    \centering    
    \caption{Template size in MB }\label{tab:tempSize}
    \begin{tabular}{c|c|c|c|c|c|c}
      \toprule 
      \multirow{2}{*}{ \makecell{HELR dimension\\ and \#Features}}  & \multicolumn{3}{c|}{Protocol~\ref{fig:shverification}}& \multicolumn{3}{c}{Protocol~\ref{fig:mal-biometric-verification}}\\  
      \cmidrule{2-4} 
      \cmidrule{5-7} 
      & P$192$ & P$224$ & P$256$ & P$192$ & P$224$ & P$256$ \\    
      \midrule 
      $16 \times 16$ , $36$ & $0.13$  & $0.15$ & $0.17$ & $0.40$ & $0.47$ & $0.53$ \\
      \midrule 
      $64 \times 64$ , $49$ & $0.72$ & $0.84$ & $0.95$ & $2.20$ & $2.56$ & $2.91$ \\
      \midrule 
      $64 \times 64$ , $94$ & $1.34$ & $1.57$ & $1.79$ & $4.11$ & $4.78$ & $5.43$ \\
      \bottomrule 
    \end{tabular}
\end{table}

Table~\ref{tab:summaryTable} summarizes the performance of our work and previous work, as reported in the respective papers, in terms of the runtime, tested modality and security strength.
It is important here to emphasize that this table is not meant to compare detailed performance numbers but indicate the order of magnitudes as each protocol was tested on different machines (i.e. mobile, laptop, workstation, etc.).   
The top half of the table reports results of protocols belonging to the first category (semi-honest client and semi-honest server) while its bottom half concerns the second category (malicious client and semi-honest server) and the third category (malicious client and malicious server).
What stands out in this table is that the reported runtimes of the existing protocols are limited to a security strength of $80$ bits; which is no longer recommended by \cite{bluekrypt}.
We show here that even when considering security levels higher than what was considered in prior work, our solution for the semi-honest model (Protocol~\ref{fig:shverification}) runs similarly fast, in the order of hundreds of milliseconds, as existing work.
Also, our Protocol~\ref{fig:mal-biometric-verification} is the only solution that achieves security in the malicious model while maintaining a similar runtime as existing semi-honest solutions in terms of order of magnitude.

\begin{table*}[h!]   
    \centering
    \begin{threeparttable}[b]  
    \caption{Performance summary of our work and previous work as reported in the respective papers}   
    \label{tab:summaryTable}
    \begin{tabular}{c|c|c|c|c|c|c|c|c|c}      
      \toprule 
      \multirow{2}{*}{\makecell{Adversarial\\model}}   & \multirow{2}{*}{Protocol} & \multicolumn{2}{c|}{Modality}  & \multirow{2}{*}{ \makecell{\#Features}} & \multirow{2}{*}{ \makecell{Classification\\method}}
        & \multirow{2}{*}{\makecell{Cryptographic\\technique} }    & \multirow{2}{*}{ \makecell{Security\\strength~\cite{bluekrypt}} } & 
        \multicolumn{2}{c}{ \makecell{\multirow{2}{*}{Runtime (s) \tnote{$\star$}}} }   \\
      \cmidrule{3-4} 
      && Supported \tnote{$\dagger$} & Tested &&&&\\
      \midrule 
      \multirow{8}{*}{\makecell{SH Client\\SH Server}} &\makecell{\cite{cheon2016ghostshell}}& Binary & Iris & $2400$ & HD\tnote{1}  & \makecell{ BGV and\\MAC}  &  80 & \multicolumn{2}{c}{ \makecell{$0.10$ (Client)\\$0.47$ (Server)}  } \\  
      \cmidrule{2-10} 
      &\makecell{\cite{karabat2015thrive}} & Binary & Face & $192$ & HD\tnote{1}  & \makecell{THE\\Xor GM}  &  80 & \multicolumn{2}{c}{  \makecell{$0.25$ (Client)\\$0.76$ (Server)} } \\  
      \cmidrule{2-10} 
      &\multirow{6}{*}{\makecell{Our work \\ (Protocol \ref{fig:shverification})} } &  \multirow{6}{*}{Real-valued} & \multirow{3}{*}{\makecell{Dynamic\\Signature}} & \multirow{3}{*}{36}  & \multirow{6}{*}{HELR\tnote{6}}   & \multirow{6}{*}{\makecell{THE\\ElGamal}} &  $96 $ & \multicolumn{2}{c}{ $0.42$ } \\ 
                                                &&&& &&&   $112$ & \multicolumn{2}{c}{ $0.48$ } \\   
                                                &&&& &&& $128$ & \multicolumn{2}{c}{$0.50$} \\ 
                                                \cmidrule{4-5} 
                                                \cmidrule{8-10} 
                                                &&& \multirow{3}{*}{Face} & \multirow{3}{*}{\makecell{$49$ \tnote{$\ast$}  and $94$ \tnote{$\ddagger$}}} &&&  $96 $ & $0.79$\tnote{$\ast$} & $0.37$ \tnote{$\ddagger$}\\
                                                &&&& &&&   $112$ & $0.82$ & $0.39$ \\   
                                                &&&& &&& $128$ & $0.88$ & $0.43$\\ 
      \midrule 
      \multirow{6}{*}{\makecell{MAL Client\\SH Server}} & \makecell{\cite{shahandashti2015reconciling}} & Binary & \makecell{User's\\profile} & $1$ & AAD\tnote{2}  & Paillier & $80$  & \multicolumn{2}{c}{ \makecell{$1.13$ (Client)\\$0.04$ (Server)} }  \\
      \cmidrule{2-10} 
      & \makecell{\cite{vsedvenka2014secure}} & Real-valued & Touchscreen & $10$ & MD\tnote{3} & \makecell{GC, DGK and\\cut-and-choose}  &  Unknown & \multicolumn{2}{c}{ $3.70$ } \\ 
      \cmidrule{2-10} 
      & \makecell{\cite{gunasinghe2017privbiomtauth}} & Real-valued & Face & $60$ & SVM\tnote{4}  & ZK-proofs &  $80$ & \multicolumn{2}{c}{ $15.42$ } \\       
      \cmidrule{2-10} 
      & \makecell{\cite{im2020practical}} & Real-valued & Face  & $128$ & ED\tnote{5}  & HE &  $80$ & \multicolumn{2}{c}{ $1.07$ }\\ 
      \midrule    

      \multirow{6}{*}{\makecell{MAL Client\\MAL Server}} & \multirow{6}{*}{  \makecell{Our work \\ (Protocol \ref{fig:mal-biometric-verification})}} & \multirow{6}{*}{Real-valued} & \multirow{3}{*}{\makecell{Dynamic\\Signature}} & \multirow{3}{*}{36}  & \multirow{6}{*}{HELR\tnote{6}}   & \multirow{6}{*}{\makecell{THE\\ElGamal\\and ZK-proofs}} &  $96 $ & \multicolumn{2}{c}{ $1.71$ } \\ 
      &&&& &&&  $112$ & \multicolumn{2}{c}{ $1.90$ } \\   
      &&&& &&&  $128$ & \multicolumn{2}{c}{ $2.01$}\\ 
      \cmidrule{4-5} 
      \cmidrule{8-10} 
      &&& \multirow{3}{*}{Face}   & \multirow{3}{*}{\makecell{$49$ \tnote{$\ast$}  and $94$ \tnote{$\ddagger$}}} &&&  $96$ & $2.10$\tnote{$\ast$} & $0.95$ \tnote{$\ddagger$}\\
      &&&& &&&   $112$ & $2.36$ & $1.17$\\   
      &&&& &&&   $128$ & $2.50$ & $1.22$\\                                           
      \bottomrule 
    \end{tabular}
    \begin{tablenotes}
      \item[1] Hamming Distance
      \item[2] Absolute Average Deviation
      \item[3] Manhattan Distance
      \item[4] Support Vector Machine 
      \item[5] Euclidean Distance  
      \item[6] Homomorphically Encrypted log Likelihood-Ratio-based classifier      
      \item[$\dagger$] Supported modality representation can be either binary feature vector or a real-valued feature vector.
      \item[$\star$] This column represents, unless specified, the overall runtime of an authentication session without considering the feature extraction. The reported numbers are determined on different systems which makes them not directly comparable but only compared in terms of the orders of magnitude.
      \item[$\ast$] PUT dataset
      \item[$\ddagger$] FRGC dataset
    \end{tablenotes}
  \end{threeparttable}  
\end{table*}

\section{Related Work}\label{section:RelatedWork} 
Over the past decades, several approaches were developed to protect biometric data, known as Biometric Template Protection (BTP)~\cite{iso_iec_24745_2011}, most of them protect it at rest when it is stored in a biometric system. 
However, the increasing amount of online services requires a more sophisticated form of biometric solution that, for a given adversarial scenario, protects the biometric data involved in the entire process.
In this direction, various biometric systems have been proposed that can be categorized according to their adversarial model.

\textbf{Semi-honest client and semi-honest server:} this category assumes that the parties are properly following the protocol steps and try to learn the biometric information only out of the exchanged messages.
To prevent such biometric information leakage, the majority of those systems perform the biometric comparison in the encrypted domain.
For instance, \cite{upmanyu2010blind} uses a support vector machine (SVM) classifier for the comparison.
This classifier yields the final score as a linear combination of the template and the probe.
The idea consists of storing the encrypted template on the server using a multiplicative homomorphic encryption scheme.
Then, for the authentication, the client sends an encrypted probe to the server who calculates the linear combination under encryption and sends it randomized to the client who decrypts it for the server. 
After cancelling out the randomization, the server learns the final score and based on it makes its decision.
To achieve the same goal, another system \cite{chun2014outsourceable} follows a hybrid approach, which is a combination of additive homomorphic encryption and garbled circuits, to perform the comparison based on either Hamming distance or Euclidean distance.
A common drawback of those systems is that a decryption is performed in the middle of the comparison process; creating a cheating opportunity for a compromised client to change the decrypted result.
Ghostshell~\cite{cheon2016ghostshell} addressed that by requiring the client to decrypt the final score along with an encrypted message authentication code (MAC) tag to ensure the integrity of the decrypted score.
The server, then, makes his decision only for valid tags. 
Another limitation of those systems is that they expose the final score to the server which leaks the closeness of a probe to the template as well as the quality of a user's probe.
For instance, a server recording the final scores of its users after a successful match, is able to determine whom of them has the most stable biometric modality.
This makes such user an interesting target to attack.

\textbf{Malicious client and semi-honest server:} this category relaxes the client's adversarial model to the malicious setting assuming that an adversary can arbitrarily deviate from the protocol not only to infer biometric information from the template but also to compromise the security of the protocol.
Unlike the previous category, the client is asked to provide proofs of each interaction he makes with the server.
For example in \cite{shahandashti2015reconciling} and \cite{gunasinghe2017privbiomtauth}, the client appends a ZK-proof to each interaction; either to prove its encrypted probe's consistency or the correctness of its computations. 
Also \cite{vsedvenka2014secure} achieves this using garbled circuits empowered by cut-and-choose technique but at the cost of the protocol's efficiency.
Besides, \cite{im2020practical} uses HE that supports quadratic operations and blinds the final score.
They prove the security of their protocol by separating two types of malicious client.
Type $1$ tries to obtain the plain values of the template and type $2$ aims at being successfully authenticated without having any biometric information of the legitimate user.
Note that in their design, the comparison decision is made by the server who delivers it to the client (representing a single user), whereas in our design this decision is made by both and known only by the client who is managing multiple users.

\textbf{Malicious client and malicious server:}
\cite{simoens2012framework} and \cite{abidin2014security} have emphasized that semi-honest security is insufficient in the context of biometric recognition since it leads to severe security risks.
Systems, belonging to both previously discussed categories, store the biometric data encrypted on the server assuming that it will follow the computations as prescribed by the protocol.
The server, indeed, will learn nothing if it follows the protocol. 
However, an inevitable question arises: given the overwhelming cyber-security threats, can we guarantee that a server, holding encrypted biometric data who is interacting with a client holding raw biometric data, will not try to deviate arbitrarily from the protocol?
Here, the threat is even worse due to cyber attacks on such servers: the server operator might not have bad intentions but can be manipulated by an external attacker. 
Thus, such servers are promising targets by themselves.
This third category confronts any arbitrary deviation coming either from the server or the client.

Thus far, THRIVE\cite{karabat2015thrive} is the only work that tried to solve the problem of providing a biometric verification protocol for both client and server in the malicious model.
While the authors made some steps forward in this direction, the resulting protocol is unfortunately not secure in the malicious model as their proof follows the definition of the semi-honest model; see Definition $2.2.1$ in~\cite{hazay2010efficient}. 
In fact, they construct a simulator that does not take the adversary's input into consideration.
This is not sufficient since a malicious adversary may substitute the corrupted party's input in order to affect the correctness of the actual output.
In the following, we explain why THRIVE does not achieve security against malicious adversaries:
It performs the biometric recognition using Hamming distance between two binary vectors, namely the probe and the template.
Basically, the server stores the template bit-wise encrypted using a $(2,2)$-THE and signed by the client.
For the authentication, the client sends the probe randomized and receives the encrypted template along with its signature.
Subsequently, the client, homomorphically and bit-wise, combines the encrypted randomness, used to randomize the probe, with the encrypted template and partially decrypts the resulted encryption.
Then, it sends to the server the encrypted randomness, the partial decryptions and a signature that binds both.
Note that this signature ensures to the server that the received messages (i.e. the encrypted randomness and the partial decryptions) were sent by the client (i.e. authenticity) and received without being modified during transmission (i.e. integrity).  
Hence, this signature does not guarantee to the server that the client has sent the correct encrypted randomness and partial decryptions as prescribed by the protocol.
In other terms, the server is not sure about how the received messages were computed.
The server, on the other hand, sends its comparison decision to the client without employing any mechanism to prove the correctness of its decision.
In this case, it could skip the last computations that lead to the comparison decision and simply send arbitrary decisions to the client.
Thus, the simulator, in the proof of Theorem $1$ provided in \cite{karabat2015thrive}, does not capture the above-mentioned malicious behaviors.
This design flaw enables to a malicious client and a malicious server to send inconsistent messages that affect the correctness of the output.
Therefore, THRIVE is not secure against malicious adversaries.

\section{Conclusion}
In this paper, we propose a biometric verification protocol secure against malicious adversaries and evaluate its performance.
We showed that the proposed HELR classifier is accurate (EER between $0.25\%$ and $0.27\%$ for faces) and diverse in terms of supported biometric modalities that can be encoded as a fixed-length real-valued feature vector. 
This approach allowed us to speed up the biometric recognition in the encrypted domain by pre-computing the classifier and organizing it into very fast lookup tables that are generated for the purpose of applying a homomorphic encryption layer.
Further, we achieved security in the malicious model by strengthening our initial protocol that is secure in the semi-honest model using ZK-proofs.
The evaluation of our protocols shows runtimes between $0.37$s and $2.50$s for security strengths varying from $96$ to $128$ bits.
This makes us achieve a stronger security level for a runtime in the order of few seconds and demonstrates that the case of malicious client and malicious server is practical.

\section*{Acknowledgement}
This work is part of the research programme KIEM with project number ENPPS.KIEM.018.001, which is jointly financed by the Dutch Research Council (NWO) and GenKey Netherlands B.V.

We would like to express our deep gratitude to Marta Gomez-Barrero, Una Kelly and Ali Khodabakhsh for providing us with the feature vectors for BMDB, PUT and FRGC databases respectively.

\bibliographystyle{IEEEtran}
\bibliography{BVSecureAgainstMaliciousAdv_Bibtex}

\end{document}